\documentclass[12pt]{article}
\usepackage{geometry}

\usepackage{soul}
\usepackage{titlesec}
\usepackage[T1]{fontenc}
\usepackage[utf8]{inputenc}
\usepackage{authblk}
\usepackage{graphics}
\providecommand{\keywords}[1]{\textbf{Keywords: } #1}
\titleformat{\section}{\centering\large\scshape}{\thesection}{1em}{}
\titleformat{\subsection}{\centering\normalsize\scshape}{\thesubsection}{1em}{}
\titleformat{\subsubsection}{\normalsize\scshape}{\thesubsubsection}{1em}{}

\usepackage{hyperref}
\hypersetup{
	colorlinks=true,
	linkcolor=blue,
	filecolor=blue,      
	urlcolor=blue,
	citecolor=blue
}

\usepackage{natbib}

\usepackage{graphicx} 
\graphicspath{{art/}}

\usepackage{setspace}
\usepackage{xr}
\usepackage{color}
\usepackage{amsmath, amssymb}
\usepackage{amsthm}
\usepackage{amsfonts,dsfont,mathrsfs}
\usepackage{booktabs,multirow,array}
\usepackage{bm}

\usepackage{algorithm}
\usepackage[noend]{algorithmic}

\makeatletter
\newenvironment{breakablealgorithm}
{
	\begin{center}
		\refstepcounter{algorithm}
		\hrule height.8pt depth0pt \kern2pt
		\renewcommand{\caption}[2][\relax]{
			{\raggedright\textbf{\ALG@name~\thealgorithm} ##2\par}%
			\ifx\relax##1\relax 
			\addcontentsline{loa}{algorithm}{\protect\numberline{\thealgorithm}##2}%
			\else 
			\addcontentsline{loa}{algorithm}{\protect\numberline{\thealgorithm}##1}%
			\fi
			\kern2pt\hrule\kern2pt
		}
	}{
		\kern2pt\hrule\relax
	\end{center}
}
\makeatother

\newcommand{\R}{\mathds{R}}
\newcommand{\E}{\mathds{E}}

\newcommand{\N}{\mathds{N}}
\newcommand{\X}{\mathds{X}}
\newcommand{\Y}{\mathds{Y}}

\renewcommand{\P}{\mathds{P}}

\newcommand{\PM}{{\rm PM}}

\newcommand{\Dcr}{\mathscr{D}}

\newcommand{\Qcr}{\mathscr{Q}}

\newcommand{\Xcr}{\mathscr{X}}

\newcommand{\Gcr}{\mathscr{G}}

\newcommand{\D}{{\rm d}}
\newcommand{\kernel}{\mathcal K}
\newcommand{\Law}{\mathcal{L}}
\newcommand{\e}{\mathrm{e}}
\newcommand{\Lap}{\mathrm{L}}

\def\simind{\stackrel{\mbox{\scriptsize{ind}}}{\sim}}
\def\simiid{\stackrel{\mbox{\scriptsize{iid}}}{\sim}}
\newcommand{\indic}{\mathds{1}}

\usepackage[normalem]{ulem}

\allowdisplaybreaks[4]

\newtheorem{theorem}{Theorem}

\newtheorem{definition}{Definition}
\newtheorem{proposition}{Proposition}
\newtheorem{corollary}{Corollary}
\newtheorem{lemma}{Lemma}

\begin{document}

\title{\scshape\LARGE{Nested Compound Random Measures}}

\author[1]{Federico Camerlenghi}
\author[1]{Riccardo Corradin}
\author[1]{Andrea Ongaro}
\affil[1]{\normalsize{Department of Economics, Management, and Statistics, University of Milano--Bicocca, 20126 Milano, Italy}}

\date{}

\maketitle 

\begin{abstract}

\noindent

Nested nonparametric processes are vectors of random probability measures widely used in the Bayesian literature to model the dependence across distinct, though related, groups of observations. These processes allow a two-level clustering, both at the observational and group levels.
Several alternatives have been proposed starting from the nested Dirichlet process by \cite{Rod08}. However, most of the available models are neither computationally efficient or mathematically tractable. In the present paper, we aim to introduce a range of nested processes that are mathematically tractable, flexible, and computationally efficient. Our proposal builds upon Compound Random Measures, which are vectors of dependent random measures early introduced by \cite{Gri17}. We provide a complete investigation of theoretical properties of our model.  In particular,  we prove a general posterior characterization for vectors of Compound Random Measures, which is interesting \textit{per se} and still not available in the current literature. Based on our theoretical results and the available posterior representation, we develop the first Ferguson \& Klass algorithm for nested nonparametric processes. We specialize our general theorems and algorithms in noteworthy examples.
We finally test the model's performance on different simulated scenarios, and we exploit the construction to study air pollution in different provinces of an Italian region (Lombardy). We empirically show how nested processes based on Compound Random Measures outperform other Bayesian competitors.

 	\vspace{12pt}
	\noindent\keywords{Bayesian nonparametrics, partial exchangeability, completely random measures, Poisson processes, nested processes}
\end{abstract}

\section{Introduction}

Exchangeability is a common assumption in the Bayesian nonparametric literature, which entails homogeneity of the observed data. However, in a large variety of applied problems data come from different, though related, studies, thus one requires more complex dependence structures. The partial exchangeability assumption \citep{Def38} is a possible solution: data are supposed to be homogeneous within the same group and conditional independent across the diverse groups. 
In this context, typical Bayesian proposals try to balance two extreme situations: the full exchangeable case and the case of unconditional independence across samples. A pioneer contribution in this direction is due to \citet{Cif78}; however, the investigation of dependent nonparametric priors has been spurred many years later by \citet{Mac99,Mac00},  with the introduction of the dependent Dirichlet process. 
Indeed, several statistical models have been proposed to induce dependence across random probability measures in the presence of multiple-sample data. Remarkable examples explore the stick-breaking construction \citep{Dun08}, the superposition of random measures \citep{Lij14,Gri13,Lij14b,Lij14c}, hierarchical structures \citep{Teh06,Cam_AOS19}, nested nonparametric processes \citep{Rod08,Cam19,Ber21,Lij23a}, thinned random measures \citep{Lau22}. 
The dependence structure induced by many of these statistical models may also be quantified in terms of Wasserstien indexes of dependence \citep{Cat21,Cat_arxiv}. 
We refer to  \citet{Qui22} for an updated review of dependent structures in Bayesian nonparametrics. 


The present paper deals with  special classes of dependent nonparametric priors, namely \textit{nested nonparametric processes}. They are extremely useful to cluster simultaneously observations  and distributions in a partially exchangeable setting. 
The first contribution for this type of models is due to \cite{Rod08}, who defined the nested Dirichlet process to cluster distributions within a Bayesian nonparametric setting. However, \cite{Cam19} have shown that the nested Dirichlet process exhibits  a degeneracy issue. More precisely, if two samples from different groups share the same  distinct value, then  the nested Dirichlet process model collapses to a situation of full exchangeability. This behavior is undesirable since the heterogeneity across samples is destroyed.  
The same degeneracy property also holds for mixture models at the level of the latent parameters. \cite{Cam19} proposed latent nested nonparametric priors to overcome this problem. Although this proposal solves the degeneracy issue of the nested Dirichlet process, it becomes hugely demanding from a computational standpoint, especially when the number of groups grows. 
Other strategies recently appeared in the Bayesian nonparametric literature to face the computational difficulties and combinatorial hurdles of latent nested models. These include a semi-hierarchical model proposed by \cite{Ber21}, where a convex linear combination of a Dirichlet process and a diffuse measure is embedded in an hierarchical structure, which allows for shared and component-specific atoms. \cite{Den21} have developed the Common Atoms Model (CAM), which are dependent Dirichlet processes whose atoms are shared across the different groups. CAM is much more computationally efficient with respect to the previous proposals, thanks to a nested slice sampler strategy. See also \citet{Ang23} and \citet{Ang24} for related contributions. Similarly, \cite{Lij23a} proposed a blending of the nested and hierarchical models. The authors focus  on the Dirichlet process case, and  overcome the degeneracy issue of the nested Dirichlet process by embedding a hierarchical structure within a nested construction. See also \cite{Wu23} for an interesting application of this model to single-cell data.


Nowadays nested processes are widely used for statistical analysis, but many nested constructions rely on the Dirichlet process \citep{Fer73}. While the use of the Dirichlet process simplifies computational procedures, it leads to a lack of flexibility compared to the range of possibilities available in Bayesian nonparametrics \citep[see, e.g.,][]{Lij07}. The present paper aims to introduce a general class of flexible nested random probability measures, whose theoretical investigation is possible as well. In addition, the theoretical analysis of this new family of priors will allow  to devise efficient computational procedures, competitive with respect to current proposals in the literature.
More precisely, we propose a class of nested priors building upon Compound Random Measures (CoRMs) by  \cite{Gri17}. CoRMs offer a very elegant solution to borrow information across different groups of observations, and they have been proved to be useful in a large number of settings.
Indeed, after the preliminary studies of \cite{Gri17}, CoRMs have been extended in several directions.  \citet{Gri18} used normalized CoRMs vectors for nonparametric regression, and they proposed a pseudo-marginal algorithm to estimate these models. \citet{Ber22} considered vectors of completely random measures where each random measure is represented through a set of latent measures, with the purpose of identifying common traits shared by sub-populations.  This set of latent measures builds upon a vector of CoRMs. \citet{Riv22} exploited CoRMs to define a novel survival regression model based on dependent random measures. 

In the present paper, we use CoRMs to define novel nested nonparametric processes. With respect to one of the most competitive alternative, i.e., the CAM of \cite{Den21}, our prior allows both for shared atoms and dependent weights across the diverse groups of data. On the contrary, in the CAM, the weights of the random probability measures are independent across groups, while the atoms are still shared. Thus, our proposal can better calibrate the desired level of interaction between the different groups by suitable specifications of the parameters. In addition, while few analytical results are available for CAM,  we provide  a full  theoretical analysis of our model. Among our results, we show a remarkable posterior representation of vectors of  CoRMs, which  is  still not available in the Bayesian literature. Such a representation is fundamental to deriving sampling schemes tailored to perform conditional posterior inference. Indeed, building upon the posterior representation, we can develop a Ferguson \& Klass algorithm \citep{Fer72} to carry out posterior inference for  mixtures of nested CoRMs. This will be the first conditional algorithm based on the Ferguson \& Klass representation for nested structures. Furthermore, from a computational perspective, the proposed Ferguson \& Klass algorithm turns out to be numerically stable and shows good performances, even when the number of groups $d$ increases (see Section~\ref{sec:nCoRMs_mixture}). 
A vector of CoRMs has a notable advantage over most competing models: it allows for the simultaneous estimation of the distribution of the entire population, which we refer to as the baseline distribution, along with the group-specific distributions. The baseline distribution can be used as a benchmark to compare group-specific profiles and better understand their unique characteristics.
We showcase the advantage of our proposal through an environmental application. Specifically, air pollution has grown largely in European countries in the last decades, and there is a large effort to act and reduce such pollution in urbanized areas. In particular, the Lombardy region, an industrialized county in the northern part of Italy, is known to be currently one of the areas in Europe with the largest amount of air pollution.  We apply the nested CoRMs model to cluster together provinces in the Lombardy area with similar pollution profiles. Such analysis helps identify clusters of homogeneous provinces that can potentially coordinate their actions to reduce the level of air pollution in the future. Additionally, it allows us to draw inferential conclusions on relevant environmental risk measures.

The paper is structured as follows. In Section~\ref{sec:main} we remind some basics on CoRMs. The posterior representation of CoRMs is the main contribution of Section~\ref{sec:post}, where we also analyze the noteworthy example of CoRMs with gamma-distributed scores and stable driven L\'evy intensity. Section~\ref{sec:nested} introduces nested CoRMs prior, along with a Bayesian analysis of the model. Computational aspects for mixture model extensions are discussed in Section~\ref{sec:comp}. The paper ends with an environmental application of the proposed model (Section~\ref{sec:pm10}). Proofs and additional illustrations are deferred to the appendix.

\section{Compound Random Measures}\label{sec:main}

We first introduce some notations. 
Let $(\Omega, \mathcal A, \P)$ be the underlying probability space, and $\X$ the Polish space where the observations take their values, supposed to be endowed with its Borel $\sigma$-algebra $\Xcr$. We also indicate by $\mathsf{P}_\X$  the space of all probability measures on $(\X, \Xcr)$, analogously  $\mathsf M_\X$ will stand  for the space of boundedly finite measures defined on $(\X, \Xcr)$, whereas $\mathcal P_\X$  and $\mathcal{M}_\X$ denote their Borel $\sigma$-algebras, respectively.\\
We assume to work in a multiple-sample framework, namely, we are provided with $d > 1$ samples coming from an array of $\X$-valued partially exchangeable observations 
$\{ X_{i,j}: \: i \geq 1, \,  j=1, \ldots , d\}$, in the spirit of \cite{Def38}. We denote by $\bm{X}_j := (X_{1,j}, \ldots , X_{n_j, j})$ the sample of size $n_j$ for the $j$th group of observations. According to the partial exchangeability assumption, there exists a vector of dependent random probability measures 
$(\tilde p_1, \ldots , \tilde p _d)$ such that 
\begin{equation} \label{eq:model}
\begin{split}
(X_{i_1, 1}, \ldots , X_{i_d, d}) \mid \tilde{p}_1, \ldots , \tilde{p}_d &\simind \tilde{p}_1 \otimes \cdots \otimes \tilde{p}_d \quad (i_1, \ldots ,i_d) \in \N^d\\
(\tilde{p}_1 , \ldots , \tilde{p}_d) & \sim \Qcr
\end{split}
\end{equation}
where $\Qcr$ is termed the de Finetti measure of the array of random variables on the space 
$\mathsf P_\X^d$, i.e., the $d$-fold product of $\mathsf P_\X$. 
The definition of the vector $(\tilde p_1, \ldots , \tilde p _d)$ is a crucial issue in the Bayesian framework. See \cite{Qui22} for a review. In this section we consider vectors of random probability measures obtained by normalizing CoRMs, according to \cite{Gri17}. 

\subsection{Definition of CoRMs}

We first recall that a  Completely Random Measure \citep[CRM,][]{Kin67} $\tilde \eta$ is a measurable map $\tilde \eta:(\Omega, \mathcal A)\to(\mathsf M_\X, \mathcal M_\X)$ which satisfies the following condition: for any disjoint sets $A_1, \dots, A_k \in \Xcr$, and for any $k \geq 1$, the random variables $\tilde \eta(A_1), \dots, \tilde \eta(A_k)$ are mutually independent. 
As proved by \cite{Kin67} \citep[see also][]{Dal08}, a CRM can be always represented as the sum of three components: (i) a part with random jumps at fixed locations; (ii) a deterministic drift; (iii) a part with random jumps and random locations. As the most of the current Bayesian nonparametric literature \citep{Lij10}, it is convenient to focus on CRMs of type (iii). In such a framework, $\tilde \eta$ equals $\sum_{i \geq 1}J_i \delta_{\tilde x_i}$, where $\{J_i\}_{i\geq 1}$ are non-negative random heights and $\{\tilde x_i\}_{i \geq 1}$ are $\X$-valued random locations. Moreover, $\tilde\eta $ can be represented as a functional of a marked Poisson point process. As a consequence, the following 
L\'evy-Khintchine representation holds true 
\begin{equation} \label{eq:LK}
\E\left[ \exp\left\{ -\int_\X  g(x) \tilde \eta(\D x)\right\}\right] = \exp\left\{ -\int_{\R_+}\int_{\X} \left(1 - \mathrm e^{-sg(x)}\right) \tilde \nu(\D s,\D x)  \right\}, 
\end{equation}
for any measurable function $g:\X\to\R$ such that $\int_\X \lvert g(x)\rvert \tilde\eta(\D x) < \infty$, almost surely. The measure $\tilde \nu$ in \eqref{eq:LK} is referred to as the  L\'evy intensity of $\tilde\eta $ and it satisfies the condition $$\int_{\R_+\times \X} \min\{s,1\}\tilde\nu(\D s, \D x) < \infty.$$ 
We restrict our attention to the case of homogeneous L\'evy intensities, i.e., $\tilde \nu$ can be factorized as follows: $\tilde \nu(\D s, \D x) = \nu(\D s) \alpha(\D x)$. The homogeneous specification is equivalent to assuming that the distribution of the jumps is independent of the distribution of the locations. We further suppose  that $\alpha$ is a finite and diffuse measure on $\X$. We finally remind that CRMs have been widely exploited in the Bayesian community to define nonparametric priors via normalization. See \citep{Reg03}, and \citep{Lij10} for a stimulating account.
CRMs are also the basic building blocks for the definition of CoRMs \citep{Gri17}.  
Vectors of CoRMs $(\tilde{\mu}_1, \ldots , \tilde{\mu}_d)$ are mainly characterized by two ingredients.
First, we need to specify a directing L\'evy measure $\nu^*$ on $\R_+$, which describes the information shared across different components of the vector $(\tilde \mu_1, \ldots, \tilde \mu_d)$ on a latent level. Secondly, we need to specify the score distribution $h$, which is a probability density function on $\R_+^d$. Such a distribution specializes the shared information into group-specific CRMs. Then, a vector of CoRMs $(\tilde{\mu}_1, \ldots , \tilde{\mu}_d)$ can be described as 
\begin{equation}\label{eq:mu_j_def}
\tilde{\mu}_j \mid \tilde{\eta}=  \sum_{i\geq 1} m_{j,i} J_i \delta_{\tilde x_i}, \quad j =1, \ldots , d,
\end{equation}
where $(m_{1,i}, \ldots , m_{d,i}) \simiid h$ are the score terms, for any $i \geq 1$, while $\tilde{\eta}= \sum_{i \geq 1} J_i \delta_{\tilde x_i}$ is a CRM on $(\X, \Xcr)$,
having L\'evy intensity $\nu^{*} (\D z ) \alpha (\D x)$, which is referred to as the driving L\'evy measure. As shown by \cite{Gri17}, a vector of CoRMs can be represented as a Completely Random Vector \citep[CRV, see, e.g.,][]{Kallenberg_2017} whose L\'evy intensity is of the form
\[
\rho_d (\D s_1, \ldots , \D s_d) = \int_{\R_+}  z^{-d}h (s_1/z, \ldots , s_d / z) \D s_1 \cdots \D s_d \nu^* (\D z).
\] 
Clearly, such a construction is pretty general, and specific choices of $\nu^*$ and $h$ lead to more tractable cases (see, e.g., Section~\ref{sec:example}). In the sequel we assume the validity of the following conditions. 
\begin{enumerate}
\item[(A1)] The directing L\'evy measure $\nu^*$ satisfies the following regularity condition
\begin{equation*}
\int_{\R_+^d} z^{-d}  \int_{\R_+} \min \{ 1, \| \bm{s}\| \} h (s_1/z, \ldots, s_d/ z) \D s_1 \cdots \D s_d \nu^* (\D z)  < \infty, 
\end{equation*}
having denoted by $\|\bm{s}\|$ the Euclidean norm of the vector $\bm{s}= (s_1, \ldots , s_d)^\intercal$.
\item[(A2)] We assume that the score distribution factorizes into independent components, with
\begin{equation*}\label{eq:prod_h}
h (s_1, \ldots , s_d) = \prod_{j=1}^d f(s_j),
\end{equation*}
where $f(\cdot)$ is a density function on $\R_+$. 
\end{enumerate}
Under the factorization of the score distribution described in \eqref{eq:prod_h}, the scores $m_{j,i}$s are independent and identically distributed (i.i.d.) according to density function $f$, for all $i \geq 1$ and for any group $j = 1, \dots, d$. A vector of dependent random measures defined as in \eqref{eq:mu_j_def} can be exploited as a building block for many statistical models, when we need to share information across different samples. 


\subsection{CoRMS for Bayesian nonparametric models}

A prior distribution $\Qcr$ in \eqref{eq:model} can be defined by normalizing a vector of  CoRMs $(\tilde\mu_1, \ldots , \tilde\mu_d)$, as suggested by \cite{Gri17}. More specifically, the vector of random probability measures $(\tilde p_1, \dots, \tilde p_d)$ in \eqref{eq:model} is obtained by divining each $\tilde\mu_j$ 
in \eqref{eq:mu_j_def} by its total mass, i.e.,  $\tilde p_j = \tilde \mu_j / \tilde \mu_j(\X)$. Such a normalization is possible if and only if $\P (0 < \tilde\mu_j (\X) < \infty) =1$, for any $j=1, \ldots , d$, whose validity is guaranteed if the following additional condition is enforced.
\begin{enumerate}
\item[(A3)] The resulting group-specific 
L\'evy measures $\nu_j$s satisfy
\begin{equation*}
\int_{\R_+}\nu_j(\D s) = \int_{\R_+}  \int_{\R_+} z^{-1}f( s/ z)\D s  \nu^* (\D z)  = \infty, 
\end{equation*}
where $f(\cdot)$ stands for the generic $j$th component of the score distribution. 
\end{enumerate}
The resulting vector of dependent random probability measures $(\tilde p_1, \ldots , \tilde p_d)$ allows to borrow information across the different groups of observations. This vector can be employed directly as a prior distribution, or convoluted with a kernel function to obtain dependent mixture models. Additionally, the normalized random measure
$\tilde p = \tilde \eta / \tilde \eta(\X)$ can be considered as a baseline distribution of the entire population. 

Scale changes of the distribution $f$ of the scores $m_{j,i}$ do not affect the joint distribution of the $\tilde p_j$s, because of normalization. Therefore, 
the mean value of the $m_{j,i}$s is immaterial and the impact of $f$ on the degree of dependence between the group-specific random probability measures lies instead in its relative variability. Specifically, the correlation between the $i$-th weights $m_{j,i} J_i$ and $m_{\ell,i} J_i$ of the generic $j$th and $\ell$th measures, with $j\neq \ell$, is equal to 
\begin{equation} \label{eq:corr}
\rho \left( m_{j,i} J_i, m_{\ell,i} J_i    \right) = \left(  1 + \frac{\sigma_m^2 }{\mu_m^2} \left( 1+ \frac{\mu_i^2}{\sigma_i^2}\right) \right)^{-1},
\end{equation}
where $\sigma_m^2$ is the variance of $m_{j,i}$, and $\mu_i$ and $\sigma_i^2$ are the mean and the variance, respectively, of $J_i$.
Therefore, for any $i$th weight and any non-degenerate distribution of $J_i$, this correlation is a decreasing function of the coefficient of variation ($\sigma_m / \mu_m$) of  $m_{j,i}$. Furthermore, by suitably specifying the coefficient of variation any degree of positive dependence can be reached.

A key object to understand the borrowing of information across samples is the  partially Exchangeable Partition Probability Function \citep[pEPPF, see e.g.][]{Cam19}, which is also a key quantity to carry out posterior inference in the dependent nonparametric framework. Consider the model \eqref{eq:model}, when  $(\tilde p_1 , \dots, \tilde p_d)$ is obtained by the normalization of a vector of CoRMs. Since the realizations of CoRMs are almost surely discrete, there can be ties within the same sample and across different samples $\bm{X}_j$, as $j=1,\ldots , d$. Thus, the $n=n_1+\cdots+n_d$ observations  may be partitioned into $K_n=k$ groups of distinct values, denoted here as $X_1^*, \ldots , X_k^*$.
Accordingly, $\bm{n}_j=(n_{1,j},\ldots,n_{k,j})$ denotes the vector of frequency counts and $n_{\ell,j}\geq 0$ is the number of elements of the $j$th sample that coincide with the $\ell$th distinct value, for any $j=1,\ldots,d$, $\ell=1,\ldots,k$, and $\sum_{j=1}^d n_{\ell,j}\ge 1$ for any $\ell=1,\ldots,k$. Note that $n_{\ell,j}=0$ means that the $\ell$th distinct value does not appear in the $j$th sample, but it appears in one of the other samples. Moreover, the $\ell$th distinct value is shared by any two samples $i$ and $j$ if and only if $n_{\ell,j}n_{\ell,i}\ge 1$. The pEPPF is then defined as
\begin{equation}
\label{eq:peppf}
\Pi_k^{(n)}(\bm{n}_1 , \cdots , \bm{n}_d) =\E\left[\int_{\X^k}\prod_{j=1}^d \prod_{\ell=1}^k \tilde{p}_j^{n_{\ell,j}} (\D x_\ell^*) \right]
\end{equation}
with the obvious constraint $\sum_{\ell=1}^k n_{\ell,j}=n_j$, for each $j=1,\ldots,d$. Equation~\eqref{eq:peppf} describes the probability distribution associated to a specific partition of the observations into distinct blocks with frequency counts $(\bm n_1, \dots, \bm n_d)$. Further, the expected value in \eqref{eq:peppf} is taken with respect to the random probability measures $(\tilde p_1, \dots, \tilde p_d)$, regardless the specific values of the atom associated to each block. Similarly to \citet{Jam09}, we can consider a set of suitable augmentation random variables $U_1, \dots, U_d$ which give us a tractable representation of the pEPPF. In particular, the generic $j$th element $U_j$ is a scale transformation of a gamma random variable $U_j = \Gamma_{n_j} / {\tilde\mu_j(\X)}$, with $\Gamma_{n_j} \sim \textsc{Gamma} (n_j, 1)$ and $\Gamma_{n_j}$ independent of $\tilde \mu_j(\X)$. Note that $U_j$ depends on the sample size $n_j$,  although this is not highlighted by our notation for the sake of simplicity. The expression of the pEPPF is provided by the following. 
\begin{proposition}\label{prop:eppf}
Suppose that  $\bm X_j:= (X_{1,j}, \ldots , X_{n_j,j})$, as $j=1, \ldots, d$, is a sample from the partially exchangeable model in \eqref{eq:model}, where $(\tilde{p}_1, \ldots ,\tilde{p}_d ) $ is a vector of normalized CoRMs. If $X_1^*, \dots, X_k^*$ are the distinct values out $(\bm X_1 ,\dots, \bm X_d)$ with frequencies $\bm n_j := (n_{1,j}, \dots, n_{k,j}), \; j = 1, \dots, d$,  then the pEPPF equals
\begin{multline}
\label{eq:EPPF}
\Pi_k^{(n)}(\bm{n}_1 , \cdots , \bm{n}_d) \\
= \int_{\R_+^d} \prod_{j=1}^d  \frac{u_j^{n_j-1}}{\Gamma (n_j)}  \exp
\left\{ -\alpha (\X)  \int_0^\infty  \left[ 1-\prod_{j=1}^d  \int_0^\infty  e^{-u_j m s} f( m) \D m  \right] \nu^*(\D s) \right\}\\
 \qquad \times \alpha^k (\X) \prod_{\ell=1}^k \int_0^\infty  \prod_{j=1}^d  \int_0^\infty e^{-u_j m s } (m s)^{n_{\ell, j}} f(m) \D  m\nu^* (\D s ) \D u_1 \cdots \D u_d .
\end{multline}
\end{proposition}
This result can also be found in Section 5 of \citep{Gri17}, where the authors describe sampling strategies for normalization of CoRMs, although this is not stated as a proposition. We provide a direct proof in Appendix~\ref{prof_prop_eppf}, because it is instrumental to derive the posterior characterization of CoRMs. We also remark that in the expression of the pEPPF one may get rid of the integrals over $u_1,, \ldots , u_d$ by disintegrating \eqref{eq:EPPF} and working conditionally on specific values of the aforementioned variables $U_1, \ldots , U_d$. 
This is convenient from a computational standpoint and leads us to deal with an augmented expression of the pEPPF.
\section{Posterior representation}\label{sec:post}
In this section we present a central result, which is still unavailable in the Bayesian literature: the posterior representation of a vector of CoRMs, conditionally on a sample from a partially exchangeable sequence as in \eqref{eq:model}. 
The posterior characterization of CoRMs is crucial to 
develop suitable conditional sampling strategies.  All the details are reported in the appendix (see Proposition~\ref{lem:discrete_comp} and Lemma~\ref{prop:prior_laplace}).

\begin{theorem} \label{thm:posterior_CORMS}
Consider the model \eqref{eq:model}, where  $(\tilde{p}_1, \ldots ,\tilde{p}_d ) \sim \Qcr$ is obtained by normalizing  a vector of CoRMs. Suppose that  $\bm X_j:= (X_{1,j}, \ldots , X_{n_j,j})$, as $j=1, \ldots, d$, are samples from the partially exchangeable model \eqref{eq:model}, and $X_1^*, \dots, X_k^*$ are the distinct values out $(\bm X_1 ,\dots, \bm X_d)$. Then, the following distributional equality holds true
\begin{equation}
\label{eq:posterior}
(\tilde{\mu}_1, \ldots, \tilde{\mu}_d) \mid \{\bm{X}_j, U_j\}_{j=1}^d \stackrel{d}{=} (\tilde{\mu}_1^\prime, \ldots, \tilde{\mu}_d^\prime) + \sum_{\ell=1}^k (T_{\ell, 1}, \ldots, T_{\ell, d}) \sigma_\ell \delta_{X_\ell^*} ,
\end{equation}
where:
\begin{itemize}
\item[(i)]  $(\tilde{\mu}_1^\prime, \ldots, \tilde{\mu}_d^\prime)$ is a vector of dependent random measures represented as
\begin{equation}\label{eq:post_measures}
\tilde{\mu}_j^\prime \mid \tilde{\eta}^\prime = \sum_{i \geq 1} m_{j,i}^\prime  J_i^\prime \delta_{\tilde{x}_i ^\prime} , \quad \tilde{\eta}^\prime = \sum_{i \geq 1} J_i^\prime \delta_{\tilde{x}_i^\prime}
\end{equation}
being $m_{j,i}^\prime \mid J_i^\prime$ independent with density $ f^\prime(m  \vert J_i^\prime) \propto e^{-U_j m J_i^\prime } f(m)$, and $\tilde{\eta} ^\prime$ is a CRM having L\'evy intensity 
\[
\nu ^\prime (\D s) \alpha (\D x) = \prod_{j=1}^d  \int_{\R_+} e^{-U_j m s} f (m) \D m \nu^* (\D s) \alpha (\D x );
\]
\item[(ii)]  the vectors of jumps $(T_{\ell, 1}, \ldots, T_{\ell, d})$, as $\ell=1, \ldots, k$, are independent having distribution $T_{\ell, j}  \mid \sigma_\ell \simind  \varphi(t_{\ell,j} \mid \sigma_\ell)$ and  $\sigma_\ell \sim \xi_\ell (\, \cdot \, )$. Specifically, $\varphi(t \mid s)\propto e^{- s U_j t} t^{n_{\ell, j}} f(t)$
is the density of $T_{\ell, j} \mid \sigma_\ell = s$ whereas 
\[
\xi_\ell(s) \propto  \prod_{j=1}^d  \int_{\R_+}  e^{-m s U_j} m^{n_{\ell, j}} f (m) \D m s^{n_{\ell, j}} \: \nu^* (\D s) 
\]
is the density of $\sigma_\ell$, for $\ell =1, \ldots, k$.
\end{itemize}
\end{theorem}

The proof of Theorem \ref{thm:posterior_CORMS} is deferred to Appendix~\ref{proof:th1}. From point (i) of Theorem \ref{thm:posterior_CORMS}, it is apparent that the posterior of the vector $(\tilde \mu_1, \dots, \tilde \mu_d)$, restricted to the set $\X^\prime =  \X \setminus \{X_1^*, \dots, X_k^*\}$, is still a CRV 
having an updated multivariate  L\'evy intensity. In addition, $(\tilde{\mu}_1^\prime, \ldots, \tilde{\mu}_d^\prime)$ resembles a vector of CoRMs, with the only exception that the score distribution depends on the weights of the underlying CRM $\tilde\eta^\prime$. This posterior representation can be considered a natural extension of the one for  NRMIs \citep[see, e.g., Theorem 1 of][]{Jam09} to the partially exchangeable setting. Further, the full conditional distributions of the random variables $U_1, \dots, U_d$ resemble the one of the NRMIs, case described in \citet{Jam09}. An explicit expression is given in Appendix~\ref{app:algo}. Finally, suitable specifications of the directing L\'evy intensity and of the score distribution lead to different types of CoRMs, and as a consequence to specific examples of Theorem \ref{thm:posterior_CORMS}.  
A relevant case of interest is studied in the subsequent section, where we concentrate on gamma scores. We refer to \citet{Gri17} for a discussion of other tractable examples.

\subsection{An example with gamma scores} \label{sec:example}

In the present section we specialize Theorem \ref{thm:posterior_CORMS} when the scores are gamma distributed, and the L\'evy measure equals the one associated with a $\sigma$-stable process, i.e.,
\begin{equation} \label{eq:choice_f_nu}
\begin{split}
f(x) & = \frac{1}{\Gamma (\phi)} \: x^{\phi-1} e^{-x}, \quad x >0,\\
\nu^* (\D z) & = \frac{\sigma \Gamma (\phi)}{\Gamma (\sigma+\phi) \Gamma (1-\sigma)} \: z^{-1-\sigma} \D z , \quad z >0,
\end{split}
\end{equation}
where $\phi >0$, $\sigma \in (0,1)$ are parameters. Thus, the induced marginal CRMs $\tilde{\mu}_j$ are $\sigma$-stable processes \citep{Gri17} with L\'evy intensity
\[
\nu_j (\D s) = \int_0^\infty  z^{-1} f (s/z) \nu^* (\D z ) \D s = \frac{s^{-1-\sigma} \sigma}{\Gamma (1-\sigma)} \D s , \quad s >0.
\]
Without loss of generality, we further suppose that the total mass of the centering measure $\alpha$ is equal to  $1$, namely $\alpha$ is a probability measure.
A model specified as in \eqref{eq:choice_f_nu} combines the flexibility of the stable driven L\'evy intensity with a tractable distribution for the scores. The parameter $\sigma$ is mainly tuning the allocation of the mass on the jumps of the driven L\'evy measure $\tilde \nu$, while the parameter $\phi$ is controlling the dependence between the probability measures $\tilde p_j$s and therefore the borrowing of information. Indeed, the coefficient of variation of $f$ is equal to $\phi^{-\frac{1}{2}}$, so that the dependence as measured by the correlation in \eqref{eq:corr} is increasing in $\phi$ and can take arbitrary positive values. 
We can investigate more deeply the effect of the parameters $\sigma$ and $\phi$ on the dependence induced by CoRMs in~\eqref{eq:choice_f_nu} by computing the expected the Kullback-Leibler divergence $\mathrm{KL}(\tilde p, \tilde p_j)$ of the generic group-specific probability measure $\tilde p_j$ from the baseline measure $\tilde p = \tilde \eta / \tilde \eta(\X)$. This can be expressed in an analytical form as a function of the parameters $\sigma$ and $\phi$, with
\begin{equation}\label{eq:KL_div}
\E \left[ \mathrm{KL} \left( \tilde p, \tilde p_j \right) \right] = - \psi (\phi) +\frac{1}{\sigma} \log\left(\frac{\Gamma (\phi +\sigma)}{\Gamma(\phi)}\right),
\end{equation} 
where $\psi$ is the digamma function. See Appendix~\ref{proof:KLdiv} for a derivation of Equation~\eqref{eq:KL_div}. For any given $\sigma$, the divergence can be shown to be a decreasing function of $\phi$, ranging from zero to infinity. This indicates that by varying $\phi$, one can obtain group-specific distributions that are arbitrarily close or distant from the baseline distribution, thus confirming the flexibility of the model. An illustration of Equation~\eqref{eq:KL_div} as function of $\phi$ is provided in Figure~\ref{fig:exp_KL}, for different values of $\sigma$.
\begin{figure}[!ht]
\centering
\includegraphics[width = 0.59\textwidth]{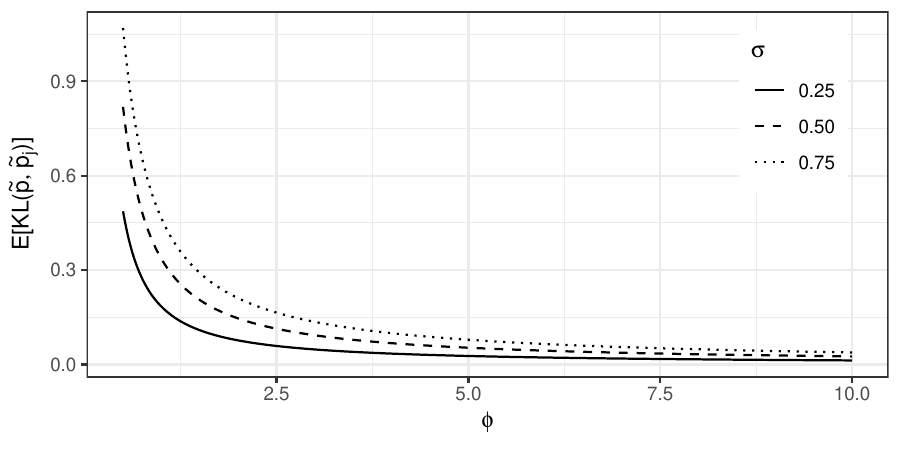}
\caption{Expected Kullback-Leibler divergence as function of $\phi$, for different values of $\sigma$.}\label{fig:exp_KL}
\end{figure}

We now apply Theorem~\ref{thm:posterior_CORMS} to obtain explicit expressions for the posterior distribution of the vector of CoRMs in the particular case under study. 
\begin{corollary} \label{cor:post_stable}
Consider the model \eqref{eq:model}, where  $(\tilde{p}_1, \ldots ,\tilde{p}_d ) \sim \Qcr$ is obtained by normalizing  a vector of CoRMs, with gamma scores and a stable directing L\'evy intensity, as in \eqref{eq:choice_f_nu}. Suppose that  $\bm X_j:= (X_{1,j}, \ldots , X_{n_j,j})$, as $j=1, \ldots, d$, are samples from the partially exchangeable model \eqref{eq:model}, and $X_1^*, \dots, X_k^*$ are the distinct values out $(\bm X_1 ,\dots, \bm X_d)$. Then, the posterior distribution of $(\tilde{\mu}_1, \ldots, \tilde{\mu}_d)$ satisfies Equation \eqref{eq:posterior}, where:
\begin{itemize}
\item[(i)]  $(\tilde{\mu}_1', \ldots, \tilde{\mu}_d')$ is a vector of dependent random measures represented as in Equation~\eqref{eq:post_measures},
being $m_{j,i}' \mid J_i'$ independent with distribution $\textsc{Gamma} (\phi, U_j J_i' +1)$, and $\tilde{\eta}'$ is a CRM having L\'evy intensity 
\[
\nu ' (\D s) \alpha (\D x) = \frac{s^{-1 - \sigma}\sigma\Gamma(\phi)}{\Gamma(\phi + \sigma) \Gamma(1 - \sigma)} \prod_{j=1}^d \frac{1}{(sU_j + 1)^\phi} \: \D s  \alpha (\D x );
\]
\item[(ii)] conditionally on the random variables $\sigma_\ell$s, the vectors of jumps $(T_{\ell, 1}, \ldots, T_{\ell, d})$, for $\ell=1, \ldots, k$, are independent and the $j$th component $T_{\ell, j}$ is distributed as a $\textsc{Gamma}(\phi +n_{\ell, j}, \sigma_\ell U_j +1) $. Besides, the distribution of $\sigma_\ell$ is characterized by the following density on $\R_+$
\begin{equation*}
	g_\ell (s)  \propto  s^{n_{\ell, \bullet}-1 - \sigma}  \prod_{j=1}^d \frac{1}{(s U_j +1)^{n_{\ell, j}+\phi}} 
	\label{eq:sigma_ell}, \qquad \ell =1, \ldots, k,
\end{equation*}
where $n_{\ell, \bullet} = \sum_{j = 1}^d n_{\ell, j}$.
\end{itemize}
\end{corollary}

We remark that a similar example with gamma scores, but with a beta directing L\'evy measure, leads to another remarkable vector of CoRMs, whose normalized marginals are Dirichlet processes \citep{Gri17}. We provide a posterior analysis of  this specific setting in Appendix~\ref{sec:post_gamma_beta}. We finally point out that,  by virtue of Corollary \ref{cor:post_stable}, we are able to generate trajectories from the posterior distribution of CoRMs, and to perform conditional algorithms for this family of models, as done for example in Section~\ref{sec:comp} on the basis of the Ferguson \& Klass representation \citep{Fer72}.

\section{Nested Compound Random Measures}\label{sec:nested}

In this section, we exploit CoRMs to define a novel class of nested nonparametric priors.
As pointed out before, these processes allow a two-level clustering, both at the observational level and at the level of the group-specific random probability measures. 
According to our proposal, the group-specific random probability measures are normalized random measures, drawn from  an almost surely discrete measure on $\mathsf M_\X$. We further suppose that the atoms of the latter measure are CoRMs. In force of the almost sure discreteness, such a specification allows us to induce ties among the random probability measures associated with the different groups of data. In addition, the resulting model ranges from a full exchangeable case, where all the groups of data are assigned to the same distribution, to a situation of full partial exchangeability across samples, where the groups of data are assigned to different, though dependent, distributions.

We first introduce nested Compound Random Measures (nCoRMs), which will be used to define a new probability $\Qcr$ for the nonparametric model in \eqref{eq:model}. For the ease of exposition, we set $S_{q-1}:= \left\{ (x_1, \ldots, x_q): \, x_i \geq 0, \text{ as } i \geq 1, \, x_1 + \dots + x_q =1 \right\}$.
\begin{definition}[nCoRMs]
A vector of random measures $(\tilde \mu_1, \dots, \tilde \mu_d)$ is said to be a vector of nCoRMs if 
\begin{equation}\label{eq:nesetd_CRM}
\tilde{\mu}_1, \ldots , \tilde{\mu}_d \mid \tilde q\simiid \tilde{q}, \quad \text{with} \quad \tilde{q}= \sum_{s =1}^q \pi_s \delta_{\mu_s},
\end{equation}
where $(\mu_1, \ldots , \mu_q)$ is a vector of CoRMs,
and  $(\pi_1, \dots, \pi_q)$, with $q \geq 1$, is a vector of random weights on the simplex $S_{q-1}.$
\end{definition}
In the sequel we assume the usual representation for the CoRMs $(\mu_1, \ldots , \mu_q)$ of Equation~\eqref{eq:mu_j_def}, 
and that (A1)--(A3) hold true. 
It is now easy to define a vector of nested random probability measures by normalizing nCoRMs and at the same time a new distribution $\Qcr$ for partially exchangeable data in \eqref{eq:model}.
nCoRMs show two main advantages with respect to the current nested proposals in the literature: (i) it is possible to easily derive the theoretical properties of these models, and then (ii) conditional posterior inference can be easily performed with new computational strategies, as we will show in Section~\ref{sec:comp}. 
Moreover, while the atoms of the different components $\mu_s$s are shared, the weights have a flexible dependence structure driven by the latent CRM $\eta$. This is a notable advantage with respect to CAM \citep{Den21}, where the weights are independent.


Note that, without loss of generality, we have assumed that $\tilde{q}$ in \eqref{eq:nesetd_CRM} is a finite-dimensional random probability measure 
in the spirit of \citet{Arg22,Lij23b}. We can possibly replace the specification of $\tilde{q}$ with an infinite-dimensional process, if required, to improve flexibility. Here we suppose that $(\pi_1, \ldots , \pi_q) \sim \Dcr_q (\beta_1, \ldots , \beta_q)$, where $\Dcr_q$ denotes the Dirichlet distribution of order $q$, with density
\begin{equation*}
f(\pi_1, \ldots , \pi_{q-1}) = \frac{\Gamma (\beta_1+\cdots +\beta_q)}{\prod_{j=1}^q \Gamma (\beta_j)} \prod_{j=1}^{q-1} \pi_j^{\beta_j-1}
(1-\pi_1- \cdots-\pi_{q-1})^{\beta_q-1},
\end{equation*}
on $S_{q-1}$. In the following sections we consider the symmetric case, that is to say $\beta_1= \cdots = \beta_q= \beta$. We further denote by $(C_1, \ldots , C_d)$ the cluster assignments of $(\tilde \mu_1, \dots, \tilde \mu_d)$, in other words $C_j =s$ if the $j$th random measure $\tilde{\mu}_j$ equals the $s$th component $\mu_s$ of $\tilde{q}$. Summing up, we are assuming that the observations  $X_{i,j}$s come from the model
\begin{equation} \label{eq:model_nested}
\begin{split}
(X_{i_1, 1}, \ldots , X_{i_d, d}) \mid \tilde{p}_1, \ldots , \tilde{p}_d &\simind \tilde{p}_1 \otimes \cdots \otimes \tilde{p}_d \quad (i_1, \ldots ,i_d) \in \N^d\\
(\tilde{p}_1 , \ldots , \tilde{p}_d) & \sim \Gcr
\end{split}
\end{equation} 
where $\Gcr$ denotes the distribution of a vector of normalized nCoRMS as described before. It is possible to derive explicitly the pEPPF for a model as in \eqref{eq:model_nested}. 
For the sake of simplicity, we focus on the case of $d=2$ groups, but the generalization to an arbitrary number of groups is straightforward. In the following, we denote by $\tau_1$ the probability that two elements sampled from $\tilde q$ coincide, i.e. $\tau_1 = \P (\tilde{\mu}_1 = \tilde{\mu}_2 ) = \sum_{s=1}^q \E[\pi_s^2]$. We further denote by $\Pi_k^{(n)}(\bm{n}_1, \bm{n}_2)$ the pEPPF in \eqref{eq:EPPF} with two groups of partially exchangeable observations, while $\Phi_k^{(n)} (\bm{n}_1+\bm{n}_2)$ stands for the Exchangeable Partition Probability Function (EPPF) in the fully exchangeable case, when all the observations come from an exchangeable model whose underlying random probability measure is a normalized CRM having  L\'evy intensity $\int_{\R_+} z^{-1} f(s/z) \nu^* (\D z) \D s$. As for the latter EPPF, 
\cite{Jam09} proved that it equals
\begin{equation*}
\begin{split}
\Phi_k^{(n)} (\bm{n}_1+\bm{n}_2) = \alpha^k (\X)\int_{\R_+}  \frac{u^{n-1}}{\Gamma (n)}   \exp \left\{  -\int_{\R_+^{2}} (1-e^{-u m s}) f(m) \D m \nu^* (\D s) \alpha (\X)\right\} \\
\times \prod_{\ell=1}^k \int_{\R_+^{2}} e^{-ums } (ms)^{n_{\ell, \bullet }}  f(m) \D m \nu^* (\D s)\: \D u .
\end{split}
\end{equation*}
We are now ready to state the following result.
\begin{theorem}\label{thm:pEPPFnested}
Let $\bm X_j = (X_{1,j}, \dots, X_{n_j,j}), \; j = 1,2$, be partially exchangeable observations as in  \eqref{eq:model_nested}, where 
$(\tilde{p}_1, \tilde{p}_2 ) \sim \Gcr$ is  a vector of normalized  nCoRMs. If  $X^*_1, \dots, X^*_k$ denote the distinct values out $( \bm X_1, \bm X_2) $ with frequency counts $\bm n_j = (n_{1,j}, \dots, n_{k,j}), \; j=1,2$, then,  the pEPPF of the random partition generated by the data, denoted by $\Psi_k^{(n)}(\bm n_1,\bm n_2)$, equals
\begin{equation}
\label{eq:pEPPF_nested}
\begin{split}
\Psi_k^{(n)}(\bm n_1,\bm n_2) = \tau_1 \Phi_k^{(n)} (\bm{n}_1+\bm{n}_2)+ (1-\tau_1)\Pi_k^{(n)}(\bm{n}_1, \bm{n}_2).
\end{split}
\end{equation}

\end{theorem}

Similar results can be found in \citet{Cam19} and \citet{Den21}. We report a direct proof of Theorem~\ref{thm:pEPPFnested} in Appendix \ref{proof:thm_nested}. It is apparent from Theorem~\ref{thm:pEPPFnested} that the probability $\tau_1$ of having $\tilde \mu_1 = \tilde \mu_2$ governs the cluster structure of the model. As far as $\tau_1$ increases, the model becomes closer to a fully exchangeable situation, where the two groups of data are homogeneous. On the other hand, when $\tau_1$ decreases to $0$, the pEPPF $\Pi_k^{(n)} (\bm{n}_1,\bm{n}_2)$  in Equation \eqref{eq:pEPPF_nested} is the dominant term, and the model both preserves heterogeneity and borrowing of information across groups.

We finally point out that, in order to obtain a more manageable expression for the pEPPF in the general setting $d\geq 2$, it is convenient to work with an augmented pEPPF, with the use of the suitable  latent label indicators $(C_1, \ldots , C_d)$ introduced before. Hence, one can work with the augmented pEPPF, defined as follows:
\begin{equation}
\label{eq:pEPPF_aug}
\E \left[ \prod_{s=1}^q \pi_s^{m_s} \right]\int_{\X^k} \E \left[  \prod_{j=1}^d \prod_{\ell =1}^k  \frac{\mu_{C_j}^{n_{\ell, j}} (\D x_\ell^*)}{\mu_{C_j}^{n_{\ell, j}} (\X )} \right],
\end{equation}
being $m_s = \sum_{j=1}^d \indic_{\{ s \}} (C_j)$, as $s=1,\ldots , q$. Note that the integral in \eqref{eq:pEPPF_aug} is nothing but the pEPPF determined in Proposition \ref{prop:eppf} for vectors of CoRMs, but conditionally on the latent label indicators $(C_1,\dots, C_d)$. The first expected value in \eqref{eq:pEPPF_aug} represents the probability distribution of the random partition governed by $\tilde q$, which can be easily written in the Dirichlet case \citep{Arg22}.

\section{Ferguson \& Klass algorithms for mixture models}\label{sec:comp}


Here we embed the model \eqref{eq:model} in a mixture setting to face density estimation and clustering. We develop suitable Ferguson \& Klass algorithms both in the case of normalized CoRMs and normalized nCoRMs. 
Let us first introduce the notation: we denote by $\kernel(\cdot \mid \cdot ):\Y \times \Theta \to \mathbb R_+$  a suitable kernel function, where $\Y$ and $\Theta$ are support and parameter spaces, respectively. 

\subsection{CoRMs mixture models}

We first consider a dependent mixture model by convoluting $\kernel(\cdot \mid \cdot )$ with respect to a vector of normalized CoRMs.
The full model specification, in its hierarchical form, is given by
\begin{equation}\label{eq:mod_mixture}
	\begin{split}
		Y_{i,j}\mid \theta_{i,j} &\simiid \mathcal \kernel(\, \cdot \, \mid \theta_{i,j}),\quad i=1,\dots,n_j, \;j=1,\dots,d\\
		\theta_{i,j} \mid \tilde p_j &\simind \tilde p_j, \quad  i=1,\dots,n_j, \;j=1,\dots,d\\
		(\tilde p_1, \dots, \tilde p_d) &\,\sim \Qcr,
	\end{split}
\end{equation} 
where $\Qcr$ is a vector of normalized CoRMs.
Such a model is useful to perform clustering analysis and density estimation in a partially exchangeable setting, borrowing information across different groups of data through the sets of latent parameters.  
In the following, we assume gamma distributed scores and a stable directing L\'evy intensity function, as in~\eqref{eq:choice_f_nu}. Here, we propose a novel Ferguson \& Klass algorithm \citep{Fer72} to address posterior inference: this is possible by virtue  of Corollary \ref{cor:post_stable}.  We remark that the  Ferguson \& Klass algorithm is a conditional method that, at each iteration, generates the trajectories of the infinite dimensional random measures. As a consequence,  uncertainty quantification is more reliable  with respect to traditional marginal methods \citep{Gri17}. The Ferguson \& Klass representation also allows to generate the weights of the CRM in a decreasing order, thus controlling the size of the neglected weights.  

\begin{breakablealgorithm}\label{algo:compound}
	\caption{Ferguson \& Klass algorithm for mixtures of CoRMs with gamma scores and stable directing L\'evy measure.}
	\begin{algorithmic}
		\STATE[{\bf 0}] Set initial values for $\varepsilon$, $\theta_{i,j}$ with $i = 1, \dots, n_j$ and $j=1,\dots,d$, $\sigma$, $\bm \sigma_{1:k}$, $\bm U_{1:d}$, $\phi$ plus the parameter of $\alpha(\cdot)$ and of its hyperpriors.

		\FOR{$r=1, \dots,  R$}
		
		\STATE[{\bf 1}] Update 
		\[
		\tilde \eta^\prime(\cdot)  \approx \sum_{i=1}^{I^\varepsilon} J_i^\prime \delta_{\tilde \theta_{i}^\prime}(\cdot)
		\]
		from its posterior distribution, exploiting the Ferguson \& Klass representation, where $\tilde \theta_i^\prime \simiid \alpha(\cdot)$, and the generic $J_i^\prime$ is obtained by inverting the L\'evy intensity at the $i$th waiting time of a standard Poisson process, and $I^{\varepsilon}$ is the largest integer $i$ such that $J_i^\prime>\varepsilon$.
		
		\STATE[{\bf 2}] Update the CRMs $(\tilde \mu_1, \dots, \tilde \mu_d)$, conditionally on $\tilde \eta^\prime$ and the distinct values  $\{\theta_\ell^*\}_{\ell = 1}^k$ out of the $\theta_{i,j}$s, with 
		\[
		\tilde{\mu}_j \approx \sum_{i=1}^{I^\varepsilon} m_{j,i}' J_i' \delta_{\tilde{\theta}_i'} + \sum_{\ell=1}^k T_{\ell, j} \sigma_\ell \delta_{\theta_\ell^*}
		\]
		where $m_{j,i}^\prime \mid - \sim \textsc{Gamma}(\phi, U_j J_i^\prime + 1)$, $T_{\ell,j}^\prime \mid - \sim \textsc{Gamma}(\phi + n_{\ell, j}, \sigma_\ell U_j  + 1)$, and the full conditional distribution of $\sigma_\ell$ on $\R_+$ is given by
		$$\xi_\ell(s\mid - ) \propto s^{n_{\ell,\bullet} - 1 - \sigma} \prod_{j=1}^d (sU_j + 1)^{-(n_{\ell,j} + \phi)}.$$
		
		\STATE[{\bf 3}] Update $\sigma$ and $\phi$ from their full conditional distributions, by performing Metropolis-Hastings steps.
		
		\STATE[{\bf 4}] Update the latent variables $\theta_{i,j}$s, where the generic $(i,j)$th element is distributed according to 
		\[
		\Law(\theta_{i,j}\mid -) \propto \sum_{r=1}^{I^\varepsilon} m_{j,r}^\prime J_i^\prime \kernel(y_{i,j}\mid \tilde \theta_r^\prime) \delta_{\tilde \theta_r^\prime}(\theta_{i,j}) + \sum_{\ell = 1}^k T_{\ell, j}\sigma_\ell \kernel(y_{i,j}\mid \theta_{\ell}^*)\delta_{\theta_{\ell}^*}(\theta_{i,j}).
		\]
		
		\STATE[{\bf 5}] Update the augmented variables $(U_1, \dots, U_d)$, with the generic $j$th element distributed as $U_j \mid - \sim \textsc{Gamma}(n_j, \tilde \mu_j(\Theta))$.
		
		\STATE[{\bf 6}] (Acceleration step) update the distinct values $\theta_\ell^*$, where 
		\[
		\Law(\D \theta_\ell^* \mid -) \propto \alpha(\D \theta_\ell^*)\prod_{j=1}^d \prod_{i \in C_{j,\ell}} \kernel(y_{i,j}\mid \theta_\ell^*).
		\]
		
		\STATE[{\bf 7}] (Hyper-acceleration step) update the parameters of $\alpha(\cdot)$.
		\ENDFOR
	\end{algorithmic}
\end{breakablealgorithm}

Algorithm \ref{algo:compound} shows the pseudo-code for generic choices of $\kernel(\cdot\mid\cdot)$ and $\alpha(\cdot)$. Steps $1-2$ are the core of our proposal: here we generate the trajectories of the vector $(\tilde \mu_1, \ldots , \tilde \mu_d)$ from its posterior distribution. In step $1$ we approximate $\tilde \eta '$ via the Ferguson \& Klass algorithm; in step $2$, conditionally on $\tilde \eta '$, one generates the trajectories of the random measures associated with each group.
Steps $6-7$ of Algorithm \ref{algo:compound} are not mandatory, but they strongly improve the mixing performance of the algorithm. We further note that for suitable choices of the kernel function $\kernel(\cdot\mid \cdot)$, the centering measure $\alpha(\cdot)$ and hyper-distribution on the parameters of $\alpha(\cdot)$, steps $6-7$ are available in a closed form. We refer to Appendix~\ref{app:algo} for a more detailed explanation of Algorithm \ref{algo:compound}, where we focused on a Gaussian kernel $\kernel(\cdot\mid \cdot)$ and the centering measure $\alpha(\cdot)$ is a normal-inverse-gamma. Algorithm \ref{algo:compound} can be generalized to other choices of the score distribution and the driven L\'evy intensity different from the ones described in \eqref{eq:choice_f_nu}, by simply adapting the first two steps.

Appendix~\ref{sec:sim_comp} shows a simulation study for the use of Algorithm~\ref{algo:compound}, in different scenarios, where we simulated groups of data from different mixtures of Gaussian distributions. From the synthetic experiments, we can appreciate that the accuracy of the density estimates improves when the sample size increases. On the other side, the point estimates of the latent partition are more reliable when the components in the data generating process are well separated. Further, in Appendix~\ref{app:sky}, we discuss an application of the model described in \eqref{eq:mod_mixture} to groups of observations coming from the Sloan Digital Sky Survey first data release \citep{Aba03}. We also provide there an estimation of the latent shared mixture model derived from the baseline distribution.

\subsection{nCoRMs mixture models} \label{sec:nCoRMs_mixture}
Similarly to \eqref{eq:mod_mixture}, we can embed a vector of nCoRMs as \eqref{eq:nesetd_CRM} in a mixture model. 
In practice, we need to replace the distribution of the mixing measures $\Qcr$ in~\eqref{eq:mod_mixture} with the distribution of a vector of normalized nCoRMs that we indicate by $\Gcr$. We can then extend the sampling strategy described in Algorithm~\ref{algo:compound}, based on the posterior representation of Section~\ref{sec:post}, to the nested case, 
with the additional advantage that we can cluster together groups of observations whose underlying mixing distributions look similar. 
More precisely, we consider Algorithm~\ref{algo:compound} with two additional  steps in order to update the cluster indicators $(C_1, \ldots , C_d)$ and the weights $(\pi_1, \ldots , \pi_q)$ of the random probability $\tilde{q}$. Algorithm~\ref{algo:nested} describes these two additional steps. Note that one can change the distribution of  $\tilde q$ by simply modifying  step $8$ of the algorithm.


\begin{breakablealgorithm}
	\caption{Ferguson \& Klass algorithm for mixtures of nCoRMs with gamma scores and stable directing L\'evy measure.}\label{algo:nested}
	\begin{algorithmic}
		\STATE[{\bf 0}] Set initial values for $\varepsilon$, $\theta_{i,j}$ with $i = 1, \dots, n_j$ and $j=1,\dots,d$, $\sigma$, $\bm \sigma_{1:k}$, $\bm U_{1:d}$, $\bm Z_{1:d}$, $\phi$ plus the parameter of $\alpha(\cdot)$ and of its hyperpriors.

		\FOR{$r=1, \dots,  R$}
		
		\STATE[{\bf 1-7}] Perform steps $1-7$ as in Algorithm \ref{algo:compound}, but conditionally on the allocation variables $C_1, \dots, C_d$. The empty components are updated according to the prior distribution.
		\STATE[{\bf 8}] Update the weights of $\tilde q$, where the full conditional of $(\pi_1, \ldots , \pi_q)$ is a Dirichlet distribution
		\[
		(\pi_1, \ldots , \pi_q) \mid-  \sim \Dcr_d (\beta+m_1, \ldots , \beta+m_q),
		\]
		with $m_s = \sum_{j = 1}^d \mathds 1_{\{s\}} (C_j)$ , for $s = 1, \dots, q$.
		\STATE[{\bf 9}] Update the cluster indicators $C_j$, $j = 1,\dots, d$, where the full conditional distribution of $C_j$ equals
		\[
		\P (C_j = s \mid -)  \propto 
		\pi_s \prod_{i=1}^{n_j} \int_\Theta  \kernel (y_{i,j} \mid \theta) \frac{\mu_s (\D \theta)}{\mu_s (\Theta)}, \quad s = 1,\dots, q.
		\]
		Note that $\mu_s$ is discrete and hence the previous integral may be easily evaluated as a sum of weights.
		\ENDFOR
	\end{algorithmic}
\end{breakablealgorithm}

We validate here the performance of nCoRMs mixture models on both clustering and density estimation through a synthetic example. Specifically, we consider different scenarios where the group-specific data generating process is a mixture of Gaussian distributions. The components of the mixtures are partially or fully shared across different groups of data, possibly with different weights. 
More specifically, we consider the following generating distributions:
\[
\begin{split}
	\bm Y_{A,i} &\simiid \frac{1}{2} \textsc N (\zeta_1, 0.6) + \frac{1}{4} \textsc N (\zeta_2, 0.6) + \frac{1}{4} \textsc N (\zeta_4, 0.6),\\
	\bm Y_{B,i} &\simiid \frac{1}{2} \textsc N (\zeta_1, 0.6) + \frac{1}{2} \textsc N (\zeta_3, 0.6),\\
	\bm Y_{C,i} &\simiid \frac{2}{5} \textsc N (\zeta_1, 0.6) + \frac{1}{5} \textsc N (\zeta_2, 0.6) + \frac{1}{5} \textsc N (\zeta_4, 0.6) + \frac{1}{5} \textsc N (\zeta_5, 0.6),
\end{split}
\]
with group-specific sample size $n_j = 100$, as $j = 1, \dots, d$. The scenarios that we analyse are described below.
\begin{enumerate}
	\item[(i)] Low number of groups, high separated components. We set $d =6$, $\bm\zeta = (0,3,6,10,15)^\intercal$, with $Y_{j,i} \sim Y_{A,i}$ for $j=1,2,3$, $Y_{j,i} \sim Y_{B,i}$ for $j=4,5$, and $Y_{j,i} \sim Y_{C,i}$ for $j=6$. 
	\item[(ii)] Low number of groups, low separated components. We set $d =6$, $\bm\zeta = (0,2,4,6.66,10)^\intercal$, with $Y_{j,i} \sim Y_{A,i}$ for $j=1,2,3$, $Y_{j,i} \sim Y_{B,i}$ for $j=4,5$, and $Y_{j,i} \sim Y_{C,i}$ for $j=6$. 
	\item[(iii)] High number of groups, high separated components. We set $d = 20$, $\bm\zeta = (0,3,6,10,15)^\intercal$, with $Y_{j,i} \sim Y_{A,i}$ for $j=1,\dots,10$, $Y_{j,i} \sim Y_{B,i}$ for $j=11,\dots,20$, and $Y_{j,i} \sim Y_{C,i}$ for $j=21,\dots,30$. 
\end{enumerate}
To mitigate possible simulation-specific distortion, all the scenarios have been replicated $100$ times and the results are averaged over the replications. 
For each scenario, we consider a model as in~\eqref{eq:mod_mixture} where the kernel function is Gaussian distribution, and $\Qcr$ is the distribution of a vector of  normalized nCoRMs. The model specification is completed by considering gamma distributed scores and stable directing L\'evy intensity, with the number of components 
$q = 6$ in scenarios (i)--(ii) and $q = 30$ in scenario (iii). Further, the base measure $\alpha(\cdot)$ equals a normal-inverse-gamma distribution $\textsc{NIG}(m_0, k_0, a_0, b_0)$. Finally, we relax the model specification by setting hyperpriors on the main parameters of the model, in particular regarding the base measure we set $m_0 \sim \textsc N(0, 10)$, $k_0 \sim \textsc{Gamma}(2, 2)$, $a_0 = 2$ and $b_0 \sim \textsc{Gamma}(2, 2)$. As for score distribution and directing L\'evy intensity we consider $\phi =1$, corresponding to a coefficient of variation for the score distribution equal to one, and $\sigma \sim \textsc{Beta}(2,2)$, where $\textsc{Beta} (a,b)$ denotes a beta distribution with parameters $(a,b)$.
We face posterior inference by exploiting Algorithm \ref{algo:nested}. For each replicate, we ran the algorithm for $15\,000$ iterations, including $10\,000$ burn-in iterations. 

In the sequel, we refer to the partition identified by the distribution label indicators $(C_1, \dots, \allowbreak C_d)$ as the \textit{nested partition}, while we refer to the partition associated to the observations simply as \textit{partition}.
Figure \ref{fig:nested_simu} summarizes the results of our simulation studies  in comparison with the current competitor, namely the CAM by \cite{Den21}. In the left panel of Figure \ref{fig:nested_simu}, we focused on the normalized variation of information distance \citep{Mei07} between the posterior point estimate of the latent partition of the data and the true one, where the point estimate is obtained using a decisional approach based on the variation of information loss function \citep{Wad18, Ras18}. The right panel shows the J-divergence, i.e., a version of the symmetrized Kullback-Leibler divergence between the true density functions and the estimated ones, averaged with respect to the different groups of data, with $J(h_1, h_2) = 0.5 KL(h_1\mid h_2)  + 0.5 KL(h_2\mid h_1)$ and $KL(h_1\mid h_2)$ denoting the Kullback-Leibler divergence between $h_1$ and $h_2$.   
It is apparent that the mixture model based on nCoRMs produces overall more accurate posterior point estimates for all the scenarios considered here. Table \ref{tab:VInested} summarizes the quality of the posterior estimates of the nested partitions: here we display the variation of information distance between the true nested partition and the posterior point estimate. 
We can appreciate that over all the scenarios considered, the nCoRMs model performs better than its competitor. In addition, the increase of the number of groups does not deteriorate the performances of the model.

\begin{figure}[h]
	\centering
	\includegraphics[width = 0.75\textwidth]{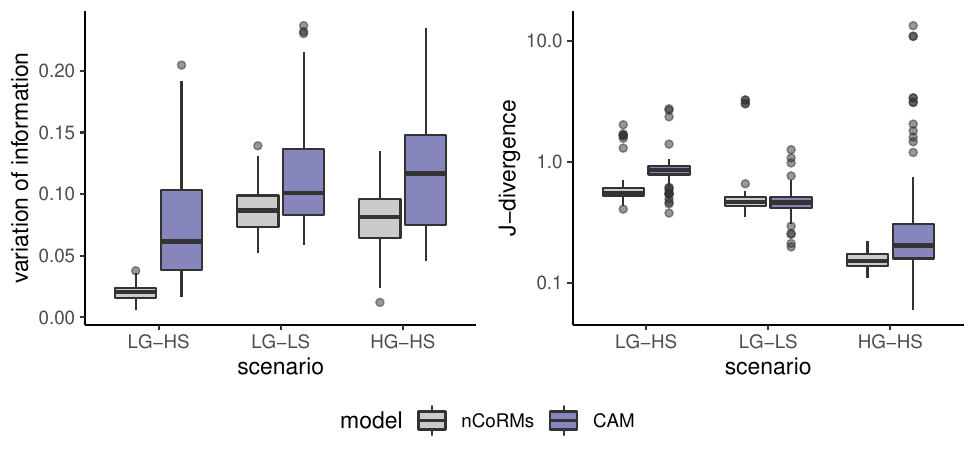}
	\caption{ 
		Left panel: variation of information distance between true and point estimate of the latent partition of the data. Right panel: J-divergence between data generating and posterior point estimate density functions. Different scenarios
		are a combination of low and high number of groups (LG and HG respectively) with low and high separation of the components (LS and HS respectively). 
	}\label{fig:nested_simu}
\end{figure}


\begin{table}[h]
	\centering
	\begin{tabular}{l|ccc}
		model & LG-HS & LG-LS & HG-HS \\ \hline
		nCoRMs & 0.000 & 0.000 & 0.001 \\
		CAM	& 0.008 & 0.014 & 0.053 \\ \hline
	\end{tabular}
	\caption{Variation of information distance between true and estimated nested partitions. The results are averaged over $100$ replications. The scenarios considered in the study
		are a combination of low and high number of groups (LG and HG respectively) with low and high separation of the components (LS and HS respectively). 
	}\label{tab:VInested}
\end{table}	

\section{Analysis of Lombardy $\PM_{10}$ data}\label{sec:pm10}


Particulate matter (PM) is the set of solid particles and liquid droplets that can be found in a gas. These particles are mainly emitted by natural sources, such as fires, soil erosion, pollens, or they can be produced by a result of the human activity, such as industries, cars and combustion in general. Here we consider the daily concentrations of $\PM_{10}$, which stands for  particles with a diameter less or equal than $10$ micrometers, in $d=12$ monitoring stations of the $12$ main towns of Lombardy region in Italy: Bergamo, Brescia, Como, Cremona, Lecco, Lodi, Mantova, Milano, Monza, Pavia, Sondrio and Varese. The observations are the daily measurements of $\PM_{10}$ in $2018$\footnote{The data are freely downloadable at \href{https://www.arpalombardia.it/Pages/Aria/Richiesta-Dati.aspx.}{https://www.arpalombardia.it/Pages/Aria/Richiesta-Dati.aspx.}}, which are indicated by $\{ Y_{i,j} : \; i=1, \ldots , n_j \}$, for monitoring station $j$, as $j=1, \ldots , d$. Further, the observations come from different provinces but in the same region. Thus, it is reasonable to suppose  that data  are partially exchangeable. This amounts to assuming that the levels of $\PM_{10}$ are exchangeable (homogeneous) within the same province and conditionally independent between the provinces. We ignore the seasonality of the data in the following analysis, instead, we focus on the whole distribution of $\PM_{10}$ to illustrate the clustering properties of nCoRMs mixture models. 

We want to investigate the distribution of $\PM_{10}$ across the different towns, possibly clustering together towns with a similar behavior. To this end, we resort to the nCoRMs mixture model as in~\eqref{eq:mod_mixture} with $\Qcr$ being the distribution of a vector of normalized  nCoRMs. We consider gamma distributed scores and stable directing intensity as in~\eqref{eq:choice_f_nu}. Given the nature of the observed variable, which takes values in the positive real line, the model is specified with a log-normal kernel function and standard distributional assumption on its parameters, i.e., with a normal-inverse-gamma base measure $\textsc{NIG}(m_0, k_0, a_0, b_0)$. In the spirit of Section~\ref{sec:nCoRMs_mixture}, we set a priori $m_0 \sim \textsc N(0, 10)$, $k_0 \sim \textsc{Gamma}(2, 2)$, $a_0 = 2$ and $b_0 \sim \textsc{Gamma}(2, 2)$, $\sigma \sim \textsc{Beta}(2,2)$ and $\phi \sim \textsc{Gamma}(2,2)$. We ran the algorithm for $15\,000$ iterations, including a burn-in period of $10\,000$. 

\begin{figure}
\centering
\includegraphics[width=0.45\textwidth]{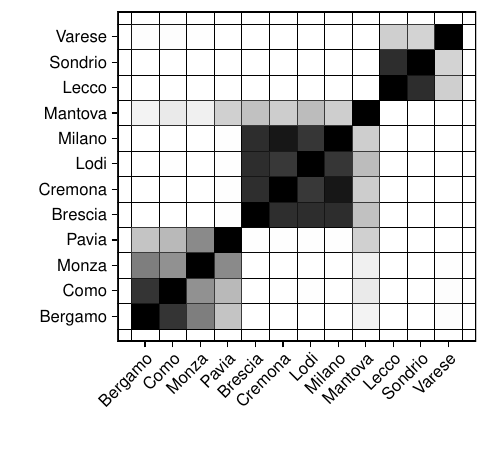}
\includegraphics[width=0.45\textwidth]{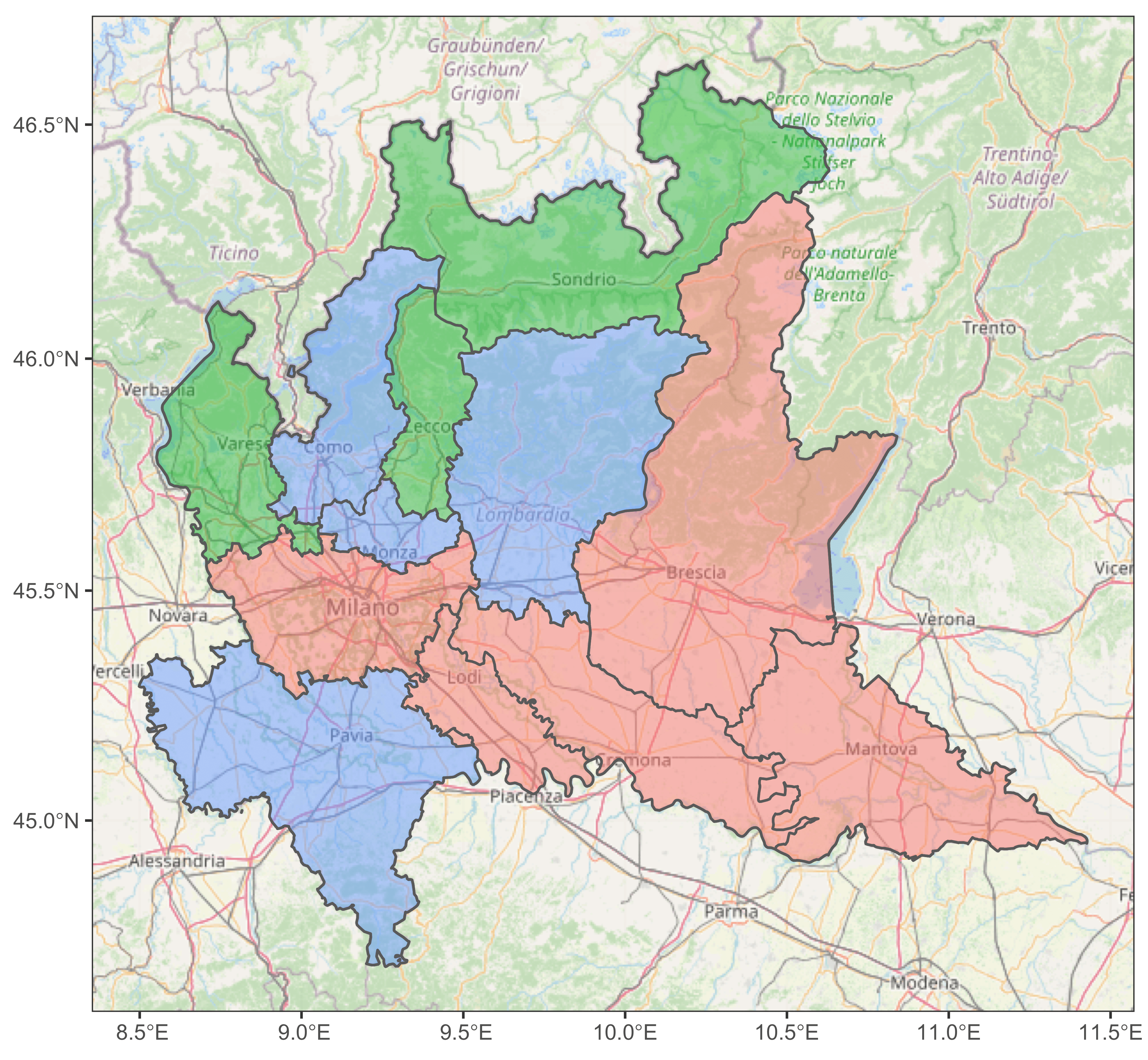}
\caption{Left panel: estimated posterior similarity of the provinces with respect to air quality data in Lombardy. Darker cells correspond to higher similarity. Right panel: geographical representation of the clusters.}\label{fig:heat_PM}
\end{figure}

We derived a point estimate for  the nested partition, which refers to the clustering of the distributions, with a strategy based on the variation of information loss function \citep{Wad18,Ras18}. We recognize three clusters which are identified by the three blocks in the left panel of Figure \ref{fig:heat_PM}, i.e., a first cluster with \{Bergamo, Como, Monza, Pavia\}, a second cluster with \{Brescia, Cremona, Lodi, Milano, Mantova\}, and a third cluster with \{Lecco, Sondrio, Varese\}. The right panel of Figure~\ref{fig:heat_PM} shows the same point estimate of the latent partition on the geographical map of the region. 

According to the European Union's regulations, the daily concentration of $\PM_{10}$ in the air should not exceed $50 \mu g/ m^3$ for more than $35$ days in  a year. 
Thanks to the available posterior trajectories of $(\tilde p_1, \ldots , \tilde p_d)$, produced via the Ferguson \& Klass algorithm, we can also perform inference on functionals of the posterior random densities. In particular, we consider the probability that the concentration of $\PM_{10}$ overcomes the aforementioned critical level of $50 \mu g/ m^3$, in formulas
\[
\tilde Q_{j,50} := \P(Y_{i,j} \geq 50 \mid \tilde p_j) = \int_{50}^\infty \int_\R\int_{\R_+} \kernel(y; \zeta, \sigma^2)\tilde p_j(\D \zeta, \D \sigma^2),
\]
as $j=1,\dots,12$, where $\zeta$ and $\sigma^2$ denote the expectation and the variance on the logarithm scale, respectively. The available posterior trajectories of $\tilde p_j$  allow  to draw samples from  the posterior distribution of $\tilde Q_{j,50}$, and also to provide a numeric approximation of its posterior distribution function.
In  Figure \ref{fig:all_dens}, we report the posterior mean of $\tilde Q_{j,50}$, denote as $\bar Q_{j,50}$,   and the corresponding $95\%$ credible intervals, jointly with the posterior density estimates of the $\PM_{10}$ concentration. It is apparent that the first cluster, \{Bergamo, Como, Monza, Pavia\}, is made of towns with a low probability of overcoming the critical level. The second cluster, \{Brescia, Cremona, Lodi, Milano, Mantova\}, is composed by cities with a high risk of overcoming the critical level of $\PM_{10}$. Finally the last cluster, \{Lecco, Sondrio, Varese\}, contains towns with a low risk of exceeding the threshold of $50 \mu g/ m^3$. Such a clustering is also coherent with the geographical morphology of the area: the cities of Lecco, Sondrio and Varese are closed to the mountains while the towns of Brescia, Cremona, Lodi, Milano and Mantova are on level ground and these areas are highly industrialized.

\begin{figure}[!ht]
\centering
\includegraphics[width=0.95\textwidth]{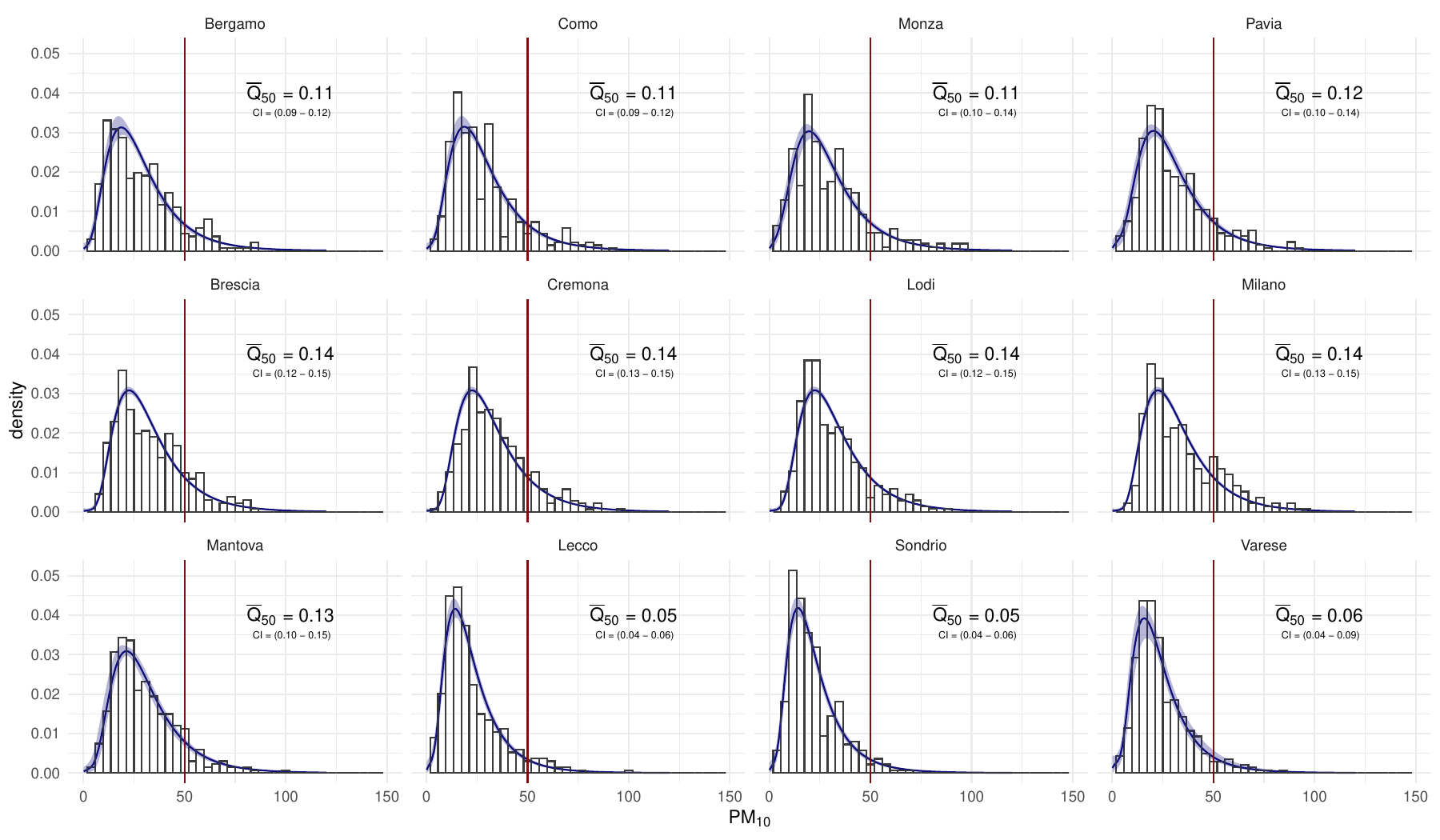}
\caption{Estimated posterior densities of $\PM_{10}$ for the different provinces  in Lombardy. The blue line corresponds to the average of the posterior random density, while the blue filled area denotes the $95\%$ credible band. The red line corresponds to the critical threshold for the $\PM_{10}$, i.e. $50 \mu g/ m^3$. In each panel is also reported $\bar Q_{j,50}$, the posterior mean of the risk associate to exceeding the critical value, along with the $95\%$ credible interval. }\label{fig:all_dens}
\end{figure}

\section*{Acknowledgements}
The first and third author are supported by the European Union -- Next Generation EU funds, component M4C2, investment 1.1., PRIN-PNRR 2022 (P2022H5WZ9). The authors gratefully acknowledge the financial support from the Italian Ministry of University and Research (MUR), ``Dipartimenti di Eccellenza'' grant 2023-2027, and the DEMS Data Science Lab for supporting this work through computational resources. The authors are grateful to Andrea Gilardi for his valuable insights on the geographical illustration in Section~\ref{sec:pm10}.

\bibliographystyle{apa}
\bibliography{bibliography}


\appendix
{\section*{\begin{LARGE}Appendices\end{LARGE}}}
	The appendix is structured as follows. Appendix~\ref{prof_prop_eppf} shows a proof of Proposition~\ref{prop:eppf}. Appendix~\ref{proof:th1} presents a proof of Theorem~\ref{thm:posterior_CORMS}. In Appendix~\ref{proof:KLdiv} we present a derivation of the expected Kullback-Leibler divergence of Equation~\ref{eq:KL_div}. Appendix \ref{proof:thm_nested} shows a direct proof of Theorem \ref{thm:pEPPFnested}. Appendix~\ref{sec:post_gamma_beta} presents the posterior characterization with gamma scores and beta directing L\'evy intensity. Appendix \ref{app:algo} describes in detail the sampling strategy of Algorithm \ref{algo:compound}. Appendix~\ref{sec:sim_comp} shows a simulation study regarding the CoRMs mixture models and Algorithm~\ref{algo:compound}. Finally, Appendix~\ref{app:sky} presents an application of the CoRMs mixture model to an astronomical dataset.

\section{Proof of Proposition \ref{prop:eppf}}\label{prof_prop_eppf}


We start from the definition of pEPPF in \eqref{eq:peppf}. An application of the Fubini-Tonelli theorem yields
\begin{equation} \label{eq:fubini}
	\E\left[\int_{\X^k}\prod_{j=1}^d \prod_{\ell=1}^k \tilde{p}_j^{n_{\ell,j}} (\D x_\ell^*) \right] =  \int_{\X^k}  \E \left[ \prod_{j=1}^d \prod_{\ell=1}^k \tilde{p}_j^{n_{\ell,j}} (\D x_\ell^*)\right].
\end{equation}
We first evaluate the integrand in \eqref{eq:fubini}.  To this end we consider $k$ disjoint balls $A_\epsilon (x_\ell^*)$, centered at $x_\ell^*$ with radius $\epsilon >0$, as $\ell=1, \ldots , k$, and we evaluate the expected value
\begin{equation*}
	\label{eq:LB}
	M_{\bm{n}_1, \ldots  , \bm{n_d}} (A_\epsilon (x_1^*) \times \ldots \times A_\epsilon (x_k^*)) := \E \left[ \prod_{j=1}^d \prod_{\ell=1}^k \tilde{p}_j^{n_{\ell,j}} (A_\epsilon (x_\ell^*))\right] .
\end{equation*}
Similarly to \citet{Jam09}, we have:
\begin{equation} \label{eq:int_gamma}
	\begin{split}
		M_{\bm{n}_1, \ldots  , \bm{n_d}} &(A_\epsilon (x_1^*) \times \ldots \times A_\epsilon (x_k^*))  =   \E \left[ \prod_{j=1}^d  \frac{1}{\tilde{\mu}_j^{n_j} (\X)}  \prod_{\ell=1}^k 
		\tilde{\mu}_j^{n_{\ell,j}} (A_\epsilon (x_\ell^*)) \right]  \\
		& = \int_{\R_+^d} \prod_{j=1}^d \frac{u_j^{n_j-1}}{\Gamma (n_j)}   \E \left[ e^{-\sum_{j=1}^d u_j \tilde{\mu}_j (\X)}
		\prod_{j=1}^d \prod_{\ell=1}^k \tilde{\mu}_j^{n_{\ell,j}} (A_\epsilon (x_\ell^*))  \right]  \D u_1 \cdots \D u_d.
	\end{split}
\end{equation}
From now on, we concentrate on the expected value in \eqref{eq:int_gamma}. 
	Then, denoting by $\X^*:= \X \setminus A_\epsilon (x_1^*) \cup \ldots \cup A_\epsilon (x_k^*)$, we get
	\begin{align*}
		\E &\left[ e^{-\sum_{j=1}^d u_j \tilde{\mu}_j (\X)}
		\prod_{j=1}^d \prod_{\ell=1}^k \tilde{\mu}_j^{n_{\ell,j}} (A_\epsilon (x_\ell^*))  \right] \\
		&\qquad = \E 
		\left[ e^{-\sum_{j=1}^d u_j \tilde{\mu}_j (\X^*)}
		\prod_{\ell=1}^k e^{-\sum_{j=1}^d u_j \tilde \mu_j (A_\epsilon (x_\ell^*))}  \prod_{j=1}^d \tilde{\mu}_j^{n_{\ell,j}} (A_\epsilon (x_\ell^*))  \right]\\
		&\qquad =  \E   \left[ e^{-\sum_{j=1}^d u_j \tilde{\mu}_j (\X^*)} \right]
		\prod_{\ell=1}^k \E \left[ e^{-\sum_{j=1}^d u_j \tilde \mu_j (A_\epsilon (x_\ell^*))}  \prod_{j=1}^d \tilde{\mu}_j^{n_{\ell,j}} (A_\epsilon (x_\ell^*))  \right]
	\end{align*}
	where the last equality follows from the independence  property of the completely random vector $\bm{\tilde \mu} :=(\tilde \mu_1, \ldots , \tilde \mu_d)$, i.e., $\bm{\tilde \mu} (A)$ and $\bm{\tilde \mu} (B)$ are independent for disjoint Borel sets $A$ and $B$. It is also easy to see that
	\begin{align*}
		& \E \left[ e^{-\sum_{j=1}^d u_j \tilde{\mu}_j (\X)}
		\prod_{j=1}^d \prod_{\ell=1}^k \tilde{\mu}_j^{n_{\ell,j}} (A_\epsilon (x_\ell^*))  \right] \\
		& \qquad =  \E   \left[ e^{-\sum_{j=1}^d u_j \tilde{\mu}_j (\X^*)} \right]
		\prod_{\ell=1}^k (-1)^{n_{\ell,1}+ \cdots + n_{\ell, d}} \E \left[ \frac{\partial^{n_{\ell, 1}+ \cdots + n_{\ell, d}}}{\partial u_1^{n_{\ell, 1}} \cdots u_d^{n_{\ell, d}}}
		e^{-\sum_{j=1}^d u_j \tilde{\mu}_j (A_\epsilon (x_\ell^*))} \right] \\
		& \qquad =  \E   \left[ e^{-\sum_{j=1}^d u_j \tilde{\mu}_j (\X^*)} \right]
		\prod_{\ell=1}^k (-1)^{n_{\ell,1}+ \cdots + n_{\ell, d}} \frac{\partial^{n_{\ell, 1}+ \cdots + n_{\ell, d}}}{\partial u_1^{n_{\ell, 1}} \cdots u_d^{n_{\ell, d}}}\E \left[ 
		e^{-\sum_{j=1}^d u_j \tilde{\mu}_j (A_\epsilon (x_\ell^*))} \right] 
	\end{align*}
	where we have exchanged derivatives and the expected value. The expectations may be easily evaluated by resorting to the explicit expression of the Laplace functional of CoRMs (see, e.g., Proposition \ref{prop:prior_laplace} with $f(m|s)=f(m)$ independent of $s$). Thus, we obtain
	\begin{align*}
		& \E \left[ e^{-\sum_{j=1}^d u_j \tilde{\mu}_j (\X)}
		\prod_{j=1}^d \prod_{\ell=1}^k \tilde{\mu}_j^{n_{\ell,j}} (A_\epsilon (x_\ell^*))  \right]  =  e^{  - \int_0^\infty  \left[ 1-\prod_{j=1}^d \int_0^\infty  e^{-u_j m s} f(m)\D m \right]  \nu^*(\D s) \alpha (\X^*) }\\
		& \qquad \times \prod_{\ell=1}^k (-1)^{n_{\ell,1}+ \cdots + n_{\ell, d}} \frac{\partial^{n_{\ell, 1}+ \cdots + n_{\ell, d}}}{\partial u_1^{n_{\ell, 1}} \cdots u_d^{n_{\ell, d}}}
		e^{  - \int_0^\infty  \left[ 1-\prod_{j=1}^d \int_0^\infty  e^{-u_j m s} f(m)\D m \right]  \nu^*(\D s) \alpha (A_\epsilon (x_\ell^*))}.
	\end{align*}
	Taking the derivatives of the previous expression it is not difficult to show that
	\begin{equation*} \label{eq:eppd_cond_u}
		\begin{split}
			& \E \left[ e^{-\sum_{j=1}^d u_j \tilde{\mu}_j (\X)}
			\prod_{j=1}^d \prod_{\ell=1}^k \tilde{\mu}_j^{n_{\ell,j}} (A_\epsilon (x_\ell^*))  \right]  \\
			& \qquad =  \exp \left\{  - \int_0^\infty  \left[ 1-\prod_{j=1}^d \int_0^\infty  e^{-u_j m s} f(m)\D m \right]  \nu^*(\D s) \right\}\\
			& \qquad \times \prod_{\ell=1}^k 
			\int_0^\infty  \prod_{j=1}^d \int_0^\infty (ms)^{n_{\ell, j}}  e^{-u_j m s} f(m)\D m   \nu^*(\D s) \alpha (A_\epsilon (x_\ell^*)) + 
			o \left( \prod_{\ell=1}^k \alpha (A_\epsilon (x_\ell^*)) \right)
		\end{split}
	\end{equation*}
	where we have used the fact that $\alpha$ is not atomic and we have neglected the higher order terms as $\epsilon \to 0$.
	We can now substitute the previous expression in \eqref{eq:int_gamma} to get
	\[
	\begin{split}
		& M_{\bm{n}_1, \ldots  , \bm{n_d}} (A_\epsilon (x_1^*) \times \ldots \times A_\epsilon (x_k^*)) \\
		& \quad =  \int_{\R_+^d} \prod_{j=1}^d \frac{u_j^{n_j-1}}{\Gamma (n_j)}   \exp \left\{  -\int_{\X} \int_0^\infty  \left[ 1-\prod_{j=1}^d \int_0^\infty  e^{-u_j m s} f(m)\D m \right]  \nu^*(\D s) \alpha (\D x) \right\}\\
		& \quad \times  \prod_{\ell=1}^k  \int_0^\infty  \prod_{j=1}^d  \int_0^\infty  e^{-u_j m s} (m s)^{n_{\ell,j}}  f(m)\D m  \nu^*(\D s) \alpha ( A_\epsilon (x_\ell^*))  \D u_1 \cdots \D u_d  + o \left( \prod_{\ell=1}^k \alpha (A_\epsilon (x_\ell^*)) \right) 
	\end{split}
	\]
	and the thesis now follows by noticing that
	\[
	\Pi_k^{(n)}(\bm{n}_1 , \cdots , \bm{n}_d) = \int_{\X^k} M_{\bm{n}_1, \ldots  , \bm{n_d}} (\D x_1 \times \cdots \times \D x_k) .
	\]

	\section{Proof of Theorem \ref{thm:posterior_CORMS}} \label{proof:th1}
	
	In order to prove Theorem \ref{thm:posterior_CORMS} we need to introduce some preliminary results. 
	First, we need to derive  the Laplace functional of a slightly different vector  $(\tilde{\mu}_1, \ldots, \tilde{\mu}_d)$  than the one considered in \eqref{eq:mu_j_def}. 
	Indeed we assume that 
	\begin{equation}\label{eq:mu_j_def_app}
		\begin{split}
			\tilde{\mu}_j \mid \tilde{\eta} &=  \sum_{i\geq 1} m_{j,i} J_i \delta_{\tilde x_i}, \quad j =1, \ldots , d\\
			\tilde{\eta} &= \sum_{i \geq 1} J_i \delta_{\tilde x_i}
		\end{split}
	\end{equation}
	where $m_{j,i}| J_i \simiid f (\, \cdot \, \vert J_i)$, i.e., the distribution of $m_{j,i}$ depends also on the jump $J_i$   of the driven CRM measure $\tilde\eta$. Such a specification generalizes the case that we considered along the manuscript. 
	\begin{proposition}\label{prop:prior_laplace}
		Let $(\tilde{\mu}_1, \ldots, \tilde{\mu}_d)$ be a vector of random measures defined as in \eqref{eq:mu_j_def_app}, where $m_{j,i}\mid J_i \simind f(\, \cdot\, \vert J_i)$. Then, the Laplace functional of $(\tilde{\mu}_1, \ldots, \tilde{\mu}_d)$ is given by
		\begin{equation*}
			\label{eq:prior_laplace}
			\Lap_{\tilde{\bm \mu}} (g_1, \ldots, g_d) =
			\exp  \left\{ -\int_\X\int_0^{+\infty} \left[ 1-  \prod_{j=1}^d  \int_0^{+\infty} \e^{-ms g_j (x)}  f(m\vert s) \D m \right]  \nu^*(\D s) \alpha (\D x) \right\}
		\end{equation*}
		for all measurable functions $g_j: \X \to \R_+$, with $j=1, \ldots , d$.
	\end{proposition}
	\begin{proof}
		By using the notation $\tilde{\mu}_j (g)= \int_\X g_j(x) \tilde{\mu}_j (\D x)$, the Laplace functional may be evaluated as follows:
		\begin{align*}
			\Lap_{\tilde{\bm \mu}} (g_1, \ldots, g_d) & = \E \left[\exp \left\{ -\tilde{\mu}_1 (g_1) - \cdots \tilde{\mu}_d (g_d) \right\}\right]
			=  \E \left[ \E \left[ \exp \left\{ -\tilde{\mu}_1 (g_1) - \cdots \tilde{\mu}_d (g_d) \right\} \Big|  \tilde{\eta} \right]\right]\\
			& =  \E \left[ \prod_{j=1}^d \prod_{i \geq 1}  \E \left[ e^{-m_{j,i} J_i  g_j (\tilde{x}_i)} \Big| \tilde{\eta }\right]\right] =
			\E \left[ \prod_{i \geq 1} \prod_{j=1}^d  \int_0^{\infty} e^{-m J_i g_j (\tilde{x}_i)} f (m\vert J_i) \D m \right]
		\end{align*}
		where we used the representation \eqref{eq:mu_j_def_app} of the random measure $\tilde{\mu}_j$, as $j=1,\ldots , d$, and the fact that $m_{j,i}\mid J_i$ are independent with distribution $f(\,\cdot\,\vert J_i)$, for any $i \geq 1$. We now exploit the fact that $\tilde{\eta} $ is a CRM with known L\'evy intensity, to obtain 
		\begin{equation*}
			\begin{split}
				\Lap_{\tilde{\bm \mu}} (g_1, \ldots, g_d)  &=
				\E \left[ \exp \left\{ \sum_{i \geq 1} \log \Big( \prod_{j=1}^d \int_{\R_+} e^{-m J_i g_j (\tilde{x}_i) }  f(m\vert J_i) \D m \Big)  \right\}\right]\\
				& = \exp \left\{ - \int_\X \int_{\R_+} \left[ 1-\prod_{j=1}^d  \int_{\R_+}  e^{-ms g_j (x)} f(m\vert s) \D m \right]  \nu^* (\D s)  \alpha (\D x)\right\}
			\end{split}
		\end{equation*}
		and the thesis follows.
	\end{proof}
	
	We also introduce the next lemma to characterize the Laplace functional of a vector of dependent random measures having random heights at $k$ fixed locations $\{x_1^*, \ldots , x_k^*\}$. The following Lemma describes an explicit expression of such a quantity, which turns out to be useful in the proof of the posterior characterization.
	\begin{lemma} \label{lem:discrete_comp}
		Let $\{x_1^*, \ldots , x_k^*\}$ be a set of $k$ deterministic values in $\X$, and consider the following random measures
		\begin{equation}\label{eq:discrete_comp}
			\tilde{\lambda}_j \mid \tilde{\lambda}_0 =  \sum_{\ell=1}^k \sigma_\ell  T_{\ell, j} \delta_{x_\ell^*}, \quad \text{as } j=1 ,\ldots , d, \quad \text{and} \quad
			\tilde{\lambda}_0 =  \sum_{\ell=1}^k \sigma_\ell \delta_{x_\ell^*}
		\end{equation}
		where $T_{\ell, j}$ and $\sigma_\ell$ are positive random variables with density on $\R_+$ given by:
		\[
		T_{\ell, j}  \mid \sigma_\ell \simind  \varphi(t_{\ell, j} \mid \sigma_\ell)  , \quad \sigma_\ell \sim \xi_\ell(\, \cdot 1\,), \quad as \;\ell=1, \ldots , k,\; and \;j=1, \ldots , d.
		\]
		Then, the Laplace functional of the vector of the random measures $(\tilde{\lambda}_1, \ldots, \tilde{\lambda}_d)$ in \eqref{eq:discrete_comp} is given by:
		\begin{equation}\label{eq:laplace_fixed}
			\Lap_{\tilde{\bm \lambda}} (g_1, \ldots, g_d) =  
			\prod_{\ell=1}^k \int_0^\infty  \prod_{j=1}^d  \int_0^\infty  e^{-\sigma t_{\ell,j} g (x_\ell^*)} \varphi( t_{\ell, j} \mid \sigma_\ell ) \D t_{\ell, j}\xi_\ell(\D \sigma_\ell) .
		\end{equation}
	\end{lemma}
	
	\begin{proof}
		A straightforward calculation shows that:
		\begin{align*}
			& \Lap_{\tilde{\bm \lambda}} (g_1, \ldots, g_d)  = 
			\E\left[  \exp \left\{ -\sum_{j=1}^d  \tilde{\lambda}_j (g_j) \right\} \right] = \E \left[\exp \left\{ - \sum_{j=1}^d  \sum_{\ell=1}^k  \sigma_\ell T_{\ell, j} g_j (x_\ell^*) \right\}\right]\\
			& \qquad = \E\left[ \E  \left[ \exp \left\{ - \sum_{j=1}^d  \sum_{\ell=1}^k  \sigma_\ell T_{\ell, j} g_j (x_\ell^*) \right\} \Big| \tilde{\lambda}_0 \right]  \right] = \E  \left[ \E \left[ \prod_{j=1}^d \prod_{\ell=1}^k  \E \left[  e^{-\sigma_\ell T_{\ell, j} g_j (x_\ell^*)}  \Big|  \tilde{\lambda}_0 \right] \right]\right] \\
			& \qquad = \prod_{\ell=1}^k  \E \left[\prod_{j=1}^d  \int_0^\infty  e^{-\sigma_\ell t g_j (x_\ell^*)}  \varphi( t_{\ell, j}\vert \sigma_\ell) \D t_{\ell, j}\right] \\
   &\qquad = \prod_{\ell=1}^k  \int_0^\infty \prod_{j=1}^d  \int_0^\infty  e^{-\sigma_\ell t g_j (x_\ell^*)}  \varphi(t_{\ell, j}\vert \sigma_\ell)  \D t_{\ell, j}\xi_\ell(\D \sigma_\ell),
		\end{align*}
		which corresponds to the right term in \eqref{eq:laplace_fixed}.
	\end{proof}

	We are now in a position to prove Theorem \ref{thm:posterior_CORMS}.
	\begin{proof}[Proof of Theorem \ref{thm:posterior_CORMS}]
		The proof requires to evaluate the posterior Laplace functional of $(\tilde \mu_1,\ldots , \tilde\mu_d)$, which equals
		\begin{equation} \label{eq:laplace_psterior_proof}
			\begin{split}
				& \E \left( \exp \{-\tilde{\mu}_1 (g_1)- \cdots - \tilde{\mu}_d (g_d) \} \Big| 
				\bm{X}_1, \ldots ,\bm{X}_d \right)\\
				&\qquad\qquad = \lim_{\epsilon \to 0} \frac{ \E \left[  \exp \{ -\tilde{\mu}_1 (g_1)- \cdots - \tilde{\mu}_d (g_d) \}   \prod_{j=1}^d \prod_{\ell=1}^k  
					\tilde{p}_j^{n_{\ell,j}} (A_\epsilon (x_\ell^*)) \right]}{\E \left[    \prod_{j=1}^d \prod_{\ell=1}^k  
					\tilde{p}_j^{n_{\ell,j}} ( A_\epsilon (x_\ell^*)) \right]} ,
			\end{split}
		\end{equation}
		for all measurable functions $g_j : \X \to \R_+$, as $j=1, \ldots , d$, and 
		$A_\epsilon (x_\ell^*)$ are again disjoint balls centered at $x_\ell^*$ with radius $\epsilon$, as $\ell=1, \ldots , k$.\\
		The denominator in \eqref{eq:laplace_psterior_proof} has been evaluated in the proof of Proposition \ref{prop:eppf} (Section \ref{prof_prop_eppf}) and it equals
		\begin{equation} \label{eq:den}
			\begin{split}
				&\int_{\R_+^d }\prod_{j=1}^d  \frac{u^{n_j-1}_j}{\Gamma (n_j)}\exp \left\{  -\alpha (\X) \int_0^\infty  \left[ 1-\prod_{j=1}^d \int_0^\infty  e^{-u_j m s} f( m) \D m \right]  \nu^*(\D s)  \right\}\\
				& \qquad \times  \prod_{\ell=1}^k  \int_0^\infty  \prod_{j=1}^d  \int_0^\infty  e^{-u_j m s} (ms)^{n_{\ell,j}}  f(m) \D m \nu^*(\D s) 
				\alpha (A_\epsilon (x_\ell^*)) \D u_1 \ldots \D u_d \\
				&\qquad+ o \left( \prod_{\ell=1}^k \alpha (A_\epsilon (x_\ell^*)) \right).
			\end{split}
		\end{equation}
		The numerator in \eqref{eq:laplace_psterior_proof} follows  along similar lines as in the proof of Proposition \ref{prop:eppf}, and it amounts to be:
		\begin{equation}
			\label{eq:num}
			\begin{split}
				&   \int_{\R_+^d }\prod_{j=1}^d  \frac{u^{n_j-1}_j}{\Gamma (n_j)}\exp\left\{ -\alpha (\X)  \int_0^\infty  \left[ 1-\prod_{j=1}^d  \int_0^\infty  e^{-(u_j+g_j (x)) m s} f( m)  \D m\right] \nu^*(\D s) \right\}\\
				& \quad \times \prod_{\ell=1}^k \int_0^\infty  \prod_{j=1}^d  \int_0^\infty e^{-(u_j +g_j (x_\ell^*))ms } (ms)^{n_{\ell, j}} f( m) \D m \nu^* (\D s )
				\alpha (A_\epsilon (x_\ell^*)) \D u_1 \ldots \D u_d\\
				&\qquad + o\left( \prod_{j=1}^k \alpha (A_\epsilon (x_\ell^*)) \right) .
			\end{split}
		\end{equation}
		If we now divide \eqref{eq:num} by \eqref{eq:den} and letting $\epsilon \to 0$, we can easily recover the posterior Laplace functional. By also conditioning on the latent variables $U_1, \ldots , U_d$, the posterior Laplace functional boils down to
		\begin{equation}  \label{eq:lap_post}
			\begin{split}
				& \exp \left\{ -\alpha (\X)  \int_0^\infty  \left[ 1- \prod_{j=1}^d  \int_0^\infty\frac{e^{-(U_j+ g_j (x)) ms} f(m) \D m}{\int_0^\infty
					e^{-U_j m s} f( m) \D m} \right] \right. \\
				& \qquad \left. \times \prod_{j=1}^d  \int_0^\infty  e^{-U_j  m s} f(m) \D m \nu^* (\D s) \right\}\\
				& \qquad  \times  \prod_{\ell=1}^k  \int_0^\infty  \frac{ \prod_{j=1}^d  \int_0^\infty e^{-(U_j +g_j (x_\ell^*))m s } (m s)^{n_{\ell, j}} f( m) \D m\nu^* (\D s )}{ \int_0^\infty  \prod_{j=1}^d  \int_0^\infty  e^{-U_j m s} (m s)^{n_{\ell,j}}  f( m) \D m\nu^*(\D s) }  .
			\end{split}
		\end{equation}
		Note that the exponential part in \eqref{eq:lap_post} is the Laplace functional of the vector $(\tilde{\mu}_1' , \ldots, \tilde{\mu}_d')$ in point i) of the theorem. Thanks to Lemma \ref{lem:discrete_comp}, it is not difficult to observe that the last product in \eqref{eq:lap_post} is the Laplace functional of the weights $T_{\ell, j}$ in point ii). We finally note that the random measures and the jumps at fixed points of discontinuities are independent random elements, and the proof is complete.
	\end{proof}
	
	\section{Proof of Equation \eqref{eq:KL_div} and its limit behaviour}\label{proof:KLdiv}
	
	Let $\tilde p_j = \tilde \mu_j / \tilde \mu_j(\X)$ be the generic $j$th dimension of a CRV and $\tilde p = \tilde \eta / \tilde \eta(\X)$ the baseline distribution. We want to characterize the expected Kullblack-Leibler divergence of $\tilde p_j$ from $\tilde p$, under the assumption of gamma-distributed scores and stable directing measure. We first recall that the previous random distributions can be written as 
	\[
	\begin{split}
		\tilde p &= \frac{1}{B} \sum_{i \geq 1} J_i \delta_{\tilde x_i}, \qquad \text{where } B=\sum_{i\geq 1} J_i,\\
		\tilde p_j &= \frac{1}{B_j} \sum_{i \geq 1} m_{j,i} J_i \delta_{\tilde x_i}, \qquad \text{where } B_j=\sum_{i\geq 1} m_{j,i} J_i .
	\end{split}
	\]
	Then, the Kullback-Leibler divergence of $\tilde p_j$ from $\tilde p$ equals
	\[
	\begin{split}
		\E \left[ \mathrm{KL} \left( \tilde p, \tilde p_j \right) \right] &= \E \left[ \sum_{i\geq 1} \frac{J_i}{B} \log \left( \frac{J_i / B}{m_{j,i} J_i / B_j} \right) \right] \\
		&= -\E \left[ \sum_{i\geq 1} \frac{J_i}{B} \log (m_{j,i}) \right] + \E \left[\sum_{i\geq 1} \frac{J_i}{B}\log \left( \frac{B_j}{B} \right) \right] \\
		&= - \sum_{i\geq 1} \E [\log(m_{j,i})] \E\left[ \frac{J_i}{B} \right] + \E \left[ \log\left(\frac{B}{B_j}\right) \sum_{i\geq 1}\frac{J_i}{B}\right] \\
		&= -\E[\log (m_{j,1})] + \E [\log (B)] - \E [\log (B_j)],
	\end{split}
	\]
 where we used the fact that all the $m_{j,i}$s are i.i.d. as $m_{j,1}$, for $j=1, \ldots , d$.
	We remark that the total mass $B$ and $B_j$ of the baseline distribution $\tilde p$ and the group-specific random probability measure $\tilde p_j$ are driven by the L\'evy intensities
	\[
	\begin{split}
		\nu^*(\D z) = \frac{\sigma \Gamma(\phi)}{\Gamma(\sigma + \phi) \Gamma(1 - \sigma)} z^{-1-\sigma}\D z, \qquad \nu_j(\D s) = \frac{\sigma}{\Gamma(1 - \sigma)} s^{-1-\sigma} \D s,
	\end{split}
	\]
	respectively. Hence, the following equality holds 
	\[
	B \stackrel{d}{=} B_j \left( \frac{\Gamma(\phi)}{\Gamma(\sigma + \phi)} \right)^{1/\sigma}
	\]
	and the Kullback-Leibler divergence becomes 
	\begin{equation}\label{eq:exp_KL_dig}
		\begin{split}
			\E \left[ \mathrm{KL} \left( \tilde p, \tilde p_j \right) \right] &=  - \E[\log(m_{j,1})] + \frac{1}{\sigma} \log \left(\frac{\Gamma(\sigma + \phi)}{\Gamma(\phi)} \right)\\
			&= - \psi(\phi) + \frac{1}{\sigma} \log(\Gamma(\sigma+\phi)) - \frac{1}{\sigma}\log(\Gamma(\phi))
		\end{split}
	\end{equation}
	where $\psi$ is the digamma function, i.e., the first derivative of the natural logarithm of the gamma function. Thus, Equation \eqref{eq:KL_div} follows. 
	
	We denote by $T_\sigma(\phi)$ the latter expression in~\eqref{eq:exp_KL_dig}, and we want to study its behaviour for a fixed value of $\sigma$, by varying $\phi$ over its support. From~\eqref{eq:exp_KL_dig}, taking the first derivative of $T^\prime_\sigma(\phi)$ with respect to $\phi$ we have
	\[
	T^\prime_\sigma(\phi) = - \psi^\prime(\phi) + \frac{\psi(\sigma + \phi) - \psi(\phi)}{\sigma}. 
	\]
	Since the digamma function is strictly increasing and concave, we have $\psi(\phi + \sigma) < \psi(\phi) + \sigma \psi^\prime(\phi)$, for $\sigma \in (0,1)$. Therefore, $T^\prime_\sigma(\phi) < 0$ for any value of $\phi$ and $T_\sigma(\phi)$ is a decreasing function of $\phi$. We recall that $\phi \in (0, +\infty)$. Regarding the behaviour of $T_\sigma(\phi)$ when $\phi$ is approaching the boundaries of its support, for $\phi \to +\infty$ we have $\psi(\phi) \approx log (\phi) + \frac{1}{2\phi}$ and $\frac{\Gamma(\sigma + \phi)}{\Gamma(\phi)} \approx \phi^\sigma$, therefore $\lim_{\phi \to +\infty} T_\sigma(\phi) = 0$.
	 
	On the other hand, for $\phi \to 0$ we have $\Gamma(\phi) \approx \frac{1}{\phi}$, thanks to the recursive property of the gamma function, $\Gamma(1 + \phi) = \phi \Gamma(\phi)$. By looking at the first derivative of the logarithm of such recursion, we have $\psi(1 + \phi) = \frac{1}{\phi} + \psi(\phi)$, which implies that $\psi(\phi) \approx - \frac{1}{\phi}$ as $\phi \to 0$. Hence, 
	\[
	T_\sigma (\phi) \approx \frac{1}{\phi} + \frac{1}{\sigma} \log(\phi)
	\]
	as $\phi \to 0$, and 
	\[
	\lim_{\phi \to 0} T_\sigma(\phi) = +\infty. 
	\]
	We can conclude that $T_\sigma(\phi)$ is a strictly decreasing function of $\phi$ in $(0, \infty)$. 
	
	\section{Proof of Theorem~\ref{thm:pEPPFnested}}\label{proof:thm_nested}
	
	We want to evaluate the pEPPF $\Psi_k^{(n)}(\bm n_1, \bm n_2)$  associated with a sample from a two-dimensional vector of nCoRMs. We have that
	
	\begin{align*}
		\begin{split}
			&\Psi_k^{(n)}(\bm n_1, \bm n_2) = \E\left[ \prod_{j=1}^2 \prod_{\ell=1}^k \tilde p_j^{n_{\ell,j}}(\D x_\ell^*) \right] = \E\left[\E\left[  \prod_{j=1}^2 \prod_{\ell=1}^k \tilde p_j^{n_{\ell,j}}(\D x_\ell^*) \Big\vert \tilde q\right] \right] \\
			&\qquad=  \E\left[\sum_{r=1}^q \pi_r\prod_{j=1}^2 \prod_{\ell=1}^k \tilde p_j^{n_{\ell,j}}(\D x_\ell^*) \delta_{p_r}(\D \tilde p_j)\right] 
		\end{split}
	\end{align*}
	where $(p_1, \dots, p_q)$ are the normalized components of $\tilde q$, i.e., $p_r= \mu_r/\mu_r (\X )$. It is straightforward to see that 
	\begin{align*}
		\E&\left[\sum_{r=1}^q\pi_r \prod_{j=1}^2 \prod_{\ell=1}^k \tilde p_j^{n_{\ell,j}}(\D x_\ell^*) \delta_{p_r}(\D \tilde p_j) \right]  \\
		&= \E\left[\sum_{r=1}^q \pi_r^2  \prod_{\ell=1}^k  p_r^{n_{\ell,\bullet}}(\D x_\ell^*) + \sum_{r\neq s}\pi_r \pi_s \prod_{\ell=1}^k  p_r^{n_{\ell,1}}(\D x_\ell^*) p_s^{n_{\ell,2}}(\D x_\ell^*) \right] \\
		&=\tau_1 \E\left[\prod_{\ell=1}^k  p_r^{n_{\ell,\bullet}}(\D x_\ell^*)  \right]  + (1 - \tau_1) \E\left[\prod_{j=1}^2\prod_{\ell=1}^k  p_j^{n_{\ell,j}}(\D x_\ell^*)  \right] 
	\end{align*}
	where $\tau_1 =  \E \left[ \sum_{r = 1}^q \pi_r^2 \right]$, and the last equation holds in force of the factorization of the score distributions in identical marginal distributions. We note that the first term corresponds to the EPPF for the fully exchangeable case of a sample from a normalized CRM with L\'evy intensity $\nu_j(\D s) = \int z^{-1} f(\D s / z) \nu^*(\D z)$, while the second term corresponds to the pEPPF of a normalized vector of CoRMs. 
	
	\section{Beta directing L\'evy measure and gamma scores}\label{sec:post_gamma_beta}
	
	In the main manuscript we focused on CoRMs with gamma scores and a  stable directing L\'evy measure. Here we concentrate on another remarkable specification of CoRMs, obtained by considering gamma scores combined with a beta directing L\'evy measure, i.e.,
	\begin{equation}\label{eq:beta_gamma}
		\begin{split}
			f(x) & = \frac{1}{\Gamma (\phi)} \: x^{\phi-1} e^{-x}, \quad x >0,\\
			\nu^* (\D z) & = z^{-1} (1 - z)^{\phi - 1} \D z , \quad 0 < z  < 1,
		\end{split}
	\end{equation}
	where $\phi > 0$. As described by~\citet{Gri17}, with the specification in~\eqref{eq:beta_gamma}, the  marginal components of $(\tilde \mu_1, \dots, \tilde \mu_d)$ are gamma processes. As a consequence,  the normalization of  $(\tilde \mu_1, \dots, \tilde \mu_d)$ results in a vector of dependent Dirichlet processes. As done in Section~\ref{sec:example}, we can now characterize the posterior distribution of $(\tilde \mu_1, \dots, \tilde \mu_d)$ with score distribution and intensity as in \eqref{eq:beta_gamma}. 
	
	\begin{corollary} \label{cor:post_beta}
		Consider the model~\eqref{eq:model}, where $(\tilde p_1, \ldots , \tilde p_d)$ is obtained by normalizing a vector of CoRMs $(\tilde \mu_1, \ldots ,\tilde \mu_d ) $ with gamma scores and a beta directing L\'evy intensity. Suppose that  $\bm{X}_j:= (X_{1,j}, \ldots , X_{n_j,j})$, as $j=1, \ldots, d$, is a sample from the model~\eqref{eq:model}, and $X_1^*, \dots, X_k^*$ are  the distinct  values out of $(\bm X_1, \ldots, \bm X_d)$. Then, the posterior distribution of $(\tilde{\mu}_1, \ldots, \tilde{\mu}_d)$ satisfies Equation \eqref{eq:posterior}, where:
		\begin{itemize}
			\item[(i)]  $(\tilde{\mu}_1', \ldots, \tilde{\mu}_d')$ is a vector of dependent random measures represented as
			\[
			\tilde{\mu}_j' \mid \tilde{\eta}' = \sum_{i \geq 1} m_{j,i}'  J_i' \delta_{\tilde{x}_i '} , \quad \tilde{\eta}' = \sum_{i \geq 1} J_i' \delta_{\tilde{x}_i'}
			\]
			being $m_{j,i}' \mid J_i'$ independent with distribution $\textsc{Gamma} (\phi, U_j J_i' +1)$, and $\tilde{\eta}'$ is a CRM having L\'evy intensity 
			\[
			\nu ' (\D s) \alpha (\D x) = s^{-1}(1 - s)^{\phi - 1} \prod_{j=1}^d \frac{1}{(1 + U_j s)^{\phi}} \: \D s  \alpha (\D x );
			\]
			\item[(ii)]  conditionally on a random variable $\sigma_\ell$, 
			the vectors of jumps $(T_{\ell, 1}, \ldots, T_{\ell, d})$, for $\ell=1, \ldots, k$, are independent and the $j$th component is distributed as a $\textsc{Gamma}
			(\phi +n_{\ell, j}, \sigma_\ell U_j +1) $. Besides, the distribution of $\sigma_\ell$ is characterized by the following density on $\R_+$
			\begin{equation*}
				g_\ell (s)  \propto  s^{n_{\ell, \bullet}-1} (1 - s)^{\phi - 1} \prod_{j=1}^d \frac{ (\phi)_{n_{\ell,j}} }{(s U_j +1)^{n_{\ell, j}+\phi}}, \qquad \ell =1, \ldots, k.
			\end{equation*}
		\end{itemize}
	\end{corollary}
	
	Corollary~\ref{cor:post_beta} can be exploited to derive a posterior sampling scheme for CoRMs with gamma scores and a beta directing L\'evy measure, in the spirit of Algorithm~\ref{algo:compound}.

\section{Details on the implementation of Algorithm~\ref{algo:compound}}\label{app:algo}

In this section, we present a detailed discussion on the implementation of Algorithm~\ref{algo:compound}, when we assume a Gaussian kernel function, gamma distributed scores and stable L\'evy directing intensity function.  
We first observe that the joint distribution of the observations $Y_{i,j}$ and the latent elements $\theta_{i,j}$, conditionally on $U_1=u_1, \dots, U_d=u_d$, is  given by
\begin{equation*}
	\begin{split}
		\left(\prod_{j=1}^d \prod_{i=1}^{n_j} \kernel (y_{i,j}\mid \theta_{i,j})\right)
		\left(\prod_{\ell=1}^k \prod_{j=1}^d  (\phi)_{n_{\ell,j}} \right)\left( \frac{\sigma \Gamma (\phi)}{\Gamma (\sigma+\phi)
			\Gamma (1-\sigma)} \right)^k \qquad\qquad\\
		\qquad\qquad\times\left(\prod_{\ell=1}^k  \int_{\R_+} \sigma_\ell^{n_{\ell, \bullet} -\sigma-1} \prod_{j=1}^d \frac{1}{(u_j \sigma_\ell+1)^{n_{\ell, j}+\phi}} \D \sigma_\ell \right)
		\left(\prod_{\ell=1}^k \alpha (\D \theta_\ell^*)\right),
	\end{split}
\end{equation*}
that can also be written in terms of the distinct values of the latent variables as follows:
\begin{equation}\label{eq:joint_law}
	\begin{split}
		\left(\prod_{j=1}^d \prod_{\ell=1}^k\prod_{i \in C_{\ell, j}} \kernel (y_{i,j}\mid \theta_\ell^*)\right)
		\left(\prod_{\ell=1}^k \prod_{j=1}^d  (\phi)_{n_{\ell,j}}\right) \left( \frac{\sigma \Gamma (\phi)}{\Gamma (\sigma+\phi)
			\Gamma (1-\sigma)} \right)^k\qquad\qquad\\ 
		\qquad\qquad\times\left(\prod_{\ell=1}^k  \int_{\R_+} \sigma_\ell^{n_{\ell, \bullet} -\sigma-1}\prod_{j=1}^d \frac{1}{(u_j \sigma_\ell+1)^{n_{\ell, j}+\phi}} \D \sigma_\ell\right) \left(\prod_{\ell=1}^k \alpha(\D \theta_\ell^*)\right),
	\end{split}
\end{equation}
where $C_{\ell, j} = \{ i : \; \theta_{i,j} = \theta^*_\ell\}$. In order to face Bayesian density estimation one can develop a conditional algorithm based on the Ferguson \& Klass representation depicted in Section \ref{sec:example} which exploits the posterior representation of Corollary \ref{cor:post_stable}. We specialize the algorithm to perform posterior inference with Gaussian kernels $\kernel(y\mid\theta) = \kernel(y\mid\zeta, \sigma^2)$, with $\theta = (\zeta, \sigma^2)$, where $\zeta$ and $\sigma^2$ denote the mean and the variance respectively. In the following, we will use $\theta$ or $(\zeta, \sigma^2)$, depending on specific needs. We further assume the centering measure $\alpha(\cdot)$ equals a normal-inverse-gamma distribution $\textsc{NIG}(m_0, k_0, a_0, b_0)$ distribution, with $m_0 \sim \textsc N(m_1, s_1^2)$, $k_0 \sim \textsc{Gamma}(a_1, b_1)$ and $b_0 \sim \textsc{Gamma}(c_1, d_1)$. The updating steps of all the variables of interest are described below.

\begin{enumerate}
	\item[{[1]}] \textbf{Update $\tilde \eta^\prime$.} We sample  $\tilde \eta^\prime$ from its distribution, described in Corollary \ref{cor:post_stable}, where 
	\[
	\tilde{\eta}' \mid - = \sum_{i \geq 1} J_i' \delta_{\tilde{\theta}_i '}\approx \sum_{i = 1}^{I^\varepsilon} J_i' \delta_{\tilde{\theta}_i '},
	\]
	for a fixed threshold $\varepsilon$ that is chosen so that $J_i' < \varepsilon$, for all $i >I_\varepsilon$. To this end, we proceed using the algorithm suggested by \citet{Fer72}, who defined a procedure to generate the weights of a CRM in a decresing order. More precisely, as $i \leq I_\varepsilon$,
	we generate the weights $J_i'$s  by the implementation of  the following steps.
	\begin{enumerate}
		\item[1.a)] Generate the $i$th waiting time of a standard Poisson process $S_i$, with $S_i - S_{i-1}\simiid \mathcal E(1)$, where $\mathcal E(1)$ denotes the exponential distribution with parameter $1$.
		\item[1.b)] Determine the jump $J_i^\prime$ by inverting the L\'evy intensity 
		\begin{equation}
			\label{eq:inversion}
			S_i = \int_{J_i'}^{+\infty} \frac{s^{-1-\sigma}}{(U_1 s+1)^\phi \cdots (U_d s +1)^\phi} \D s \: \frac{\sigma \Gamma (\phi)}{\Gamma (\sigma +\phi) \Gamma (1-\sigma )}
		\end{equation}
		as $i \geq 1$. We remark that the calculation of the integral in \eqref{eq:inversion} can be simplified by considering a change of variable $\gamma: \mathbb R_+ \to (0,1)$, with $y = \gamma(s) = \frac{2}{\pi} \arctan(s)$,
		and we obtain
		\[
		\begin{split}
			S_i &= \frac{\sigma\Gamma(\phi)}{\Gamma(\sigma+\phi)\Gamma(1-\sigma)}\\
			&\times \int_{\frac{2}{\pi}\arctan(J_i^\prime) }^1 \frac{\left[\tan(\frac{\pi}{2}y)\right]^{-1-\sigma}}{\left[U_1\tan(\frac{\pi}{2}y)+1\right]^\phi\cdots\left[U_d\tan(\frac{\pi}{2}y)+1\right]^\phi}  \frac{\pi\sec^2(y\tfrac{\pi}{2})}{2}dy.
		\end{split}
		\]
		We note that the limits of  the integral on the right hand side are in between $[0,1]$, thus the equation  is more manageable and it can be solved numerically to determine $J_i'$.
		\item[1.c)] Generate the atoms $\tilde\theta_1^\prime, \dots, \tilde\theta_{I^\epsilon}^\prime$ of the CRM $\tilde \eta^\prime$, where $\tilde{\theta}_i '\sim \alpha(\cdot)$.
	\end{enumerate}
	\item[{[2]}] \textbf{Update $(\tilde \mu_1, \dots, \tilde \mu_d)$.} Conditionally on $\tilde \eta^\prime$ and the fixed locations $\{\theta_\ell^*\}_{\ell = 1}^k$, the posterior trajectory of  $\tilde \mu_j$, as $j=1, \ldots ,d$, can be approximated as follows
	\[
	\tilde{\mu}_j \mid  - \approx \sum_{i=1}^{I^\varepsilon} m_{j,i}' J_i' \delta_{\tilde{\theta}_i'} + \sum_{\ell=1}^k T_{\ell, j} \sigma_\ell \delta_{\theta_\ell^*},
	\]
	where, according to Corollary \ref{cor:post_stable}, we have to implement the subsequent steps. 
	\begin{enumerate}
		\item[2.a)] Generate $m_{j,i}^\prime \mid - \sim \textsc{Gamma}(\phi, U_j J_i^\prime + 1)$.
		\item[2.b)] Generate $\sigma_\ell$, $\ell = 1, \dots, k$, whose density function is given by
		\[
		g_\ell(s\mid - ) \propto s^{n_\ell - \sigma - 1} \prod_{j=1}^d (sU_j + 1)^{-(n_{\ell,j} + \phi)}.
		\]
		Note that we can easily sample from $g_\ell (\, \cdot \, \mid - )$ exploiting an importance sampling strategy, i.e., we sample a set of $R$ temporary elements $\{\sigma_{\ell,1}^{(t)}, \dots, \sigma_{\ell,R}^{(t)}\}$ from a $\textsc{Gamma}(a_\sigma, b_\sigma)$ distribution, and then we update $\sigma_\ell$ by a random choice from the following discrete probability distribution
		\[
		P(\sigma_\ell = \sigma_{\ell,r}^{(t)}\mid - ) \propto \frac{\big(\sigma_{\ell,r}^{(t)}\big)^{n_\ell - \sigma - 1}\Gamma(a_\sigma)}{b_\sigma^{a_\sigma}\sigma_{\ell,r}^{(t)} e^{-b_\sigma \sigma_{\ell,r}^{(t)}}} \prod_{j=1}^d \big(\sigma_{\ell,r}^{(t)}U_j + 1\big)^{-(n_{\ell,j} + \phi)},\qquad r=1, \ldots , R.
		\].
		\item[2.c)] Generate $T_{\ell,j}^\prime \mid - \sim \textsc{Gamma}(\phi + n_{\ell, j}, \sigma_\ell U_j  + 1)$.
	\end{enumerate}
	\item[{[3]}] \textbf{Update $\sigma$ and $\phi$.} From \eqref{eq:joint_law} we can recover the full conditional distributions  of $\sigma$ and $\phi$:
	\[
	\Law(\sigma \mid - ) \propto \left[ \frac{\sigma \Gamma (\phi)}{\Gamma (\sigma+\phi)
		\Gamma (1-\sigma)} \right]^k \prod_{\ell=1}^k  \int_{\R_+} \sigma_\ell^{n_{\ell, \bullet} -\sigma-1} \prod_{j=1}^d \frac{1}{(U_j \sigma_\ell+1)^{n_{\ell, j}+\phi}} \D \sigma_\ell
	\]
	and
	\[
	\Law(\phi \mid - ) \propto \prod_{\ell=1}^k\prod_{j=1}^d(\phi)_{n_{\ell,j}}\left[ \frac{\Gamma (\phi)}{\Gamma (\sigma+\phi)
	} \right]^k \prod_{\ell=1}^k  \int_{\R_+} \sigma_\ell^{n_{\ell, \bullet} -\sigma-1} \prod_{j=1}^d \frac{1}{(U_j \sigma_\ell+1)^{n_{\ell, j}+\phi}} \D \sigma_\ell,
	\]
	where the integrals can be easily evaluated with the change of variable described at~1.b. For the update of $\sigma$ and $\phi$, we perform a Metropolis-Hastings step by sampling the  transformed parameters $(\omega, \rho)$, where $\omega = \tan(\pi\sigma - \pi/ 2)$ and $\rho = \log(\phi)$. 
	\item[{[4]}] \textbf{Update the latent variables $\theta_{i,j}$s.} We sample $\theta_{i,j}$ from  the following distribution
	\[
	\Law(\theta_{i,j}\mid -) \propto \sum_{r=1}^{I^\varepsilon} m_{j,r}^\prime J_i^\prime \kernel\big(y_{i,j}\mid \tilde \theta_i^\prime\big) \delta_{\tilde \theta_i^\prime}(\theta_{i,j}) + \sum_{\ell = 1}^k T_{\ell, j}\sigma_\ell \kernel\big(y_{i,j}\mid \theta_{\ell}^*\big)\delta_{\theta_{\ell}^*}(\theta_{i,j}),
	\]
	where $\kernel$ is a Gaussian kernel.
	\item[{[5]}] \textbf{Update $(U_1, \dots, U_d)$.} Since we are working conditionally on $\tilde{\mu}_j$, as $j=1,\ldots , d$, one may exploit \eqref{eq:int_gamma} to observe that the $U_j$s are independent with the following gamma distribution
	\begin{equation*}
		U_j \mid - \sim {\textsc{Gamma}} (n_j, \tilde{\mu}_j (\Theta)), \qquad j=1, \ldots, d.
	\end{equation*}
	\item[{[6]}] \textbf{Resample the distinct values $\theta_\ell^* = (\zeta_\ell^*, \sigma_\ell^{2*})$.} With the previous distributional assumptions, we have
	\[
	\Law(\D \zeta_\ell^*, \D \sigma_\ell^{2*}\mid -) \propto \alpha(\D \zeta_\ell^*, \D \sigma_\ell^{2*}) \prod_{j=1}^d \prod_{i \in C_{j,\ell}} \kernel(y_{i,j}\mid \zeta_\ell^*, \sigma_\ell^{2*}),
	\]
	where $\alpha$ stands for the probability law of a normal-inverse-gamma random variable. We set $\bar y_{\ell} = \frac{1}{n_{\ell, \bullet}}\sum_{j=1}^d \sum_{i \in C_{\ell, j}} y_{i,j} $, with $n_{\ell, \bullet} = \sum_{j=1}^d \lvert C_{\ell, j}\rvert$ the total number of latent parameters $(\zeta_{j,i}, \sigma_{j,i}^2)$ equal to the $\ell$th distinct value $ (\zeta_\ell^*, \sigma_\ell^{2*})$. Standard computations show that
	\[
	\begin{split}
		\sigma_{\ell}^{2*}\mid - &\sim \textsc{Inv-Gamma}(a_0^*, b_0^*),\\
		\zeta_\ell^*\mid \sigma_{\ell}^{2*}, - &\sim \mathrm N(m_{0}^*, \sigma_{\ell}^{2*} / k_0^*),
	\end{split}
	\]
	where the parameters of the full conditional distributions of $\sigma_{\ell}^{2*}$ and $\zeta_{\ell}^{*} \mid \sigma_{\ell}^{2*}$ are
	\[
	\begin{split}
		m_0^* &= \frac{1}{k_0 + n_{\ell,\bullet}} \left( k_0 m_0 + n_{\ell, \bullet} \bar y_{\ell}\right),\\
		k_0^* &= k_0 + n_{\ell, \bullet},\\
		a_0^* &= a_0 + \frac{n_{\ell, \bullet}}{2},\\
		b_0^* &= b_0 + \frac{1}{2}\left[ \sum_{j=1}^d\sum_{i \in C_{\ell, j}} (y_{i,j} - \bar y_{\ell})^2 + \frac{n_{j,\bullet} k_0(\bar y - m_0)^2}{k_0 + n_{j,\bullet}} \right].
	\end{split}
	\]
	\item[{[7]}] \textbf{Update the parameters of $\alpha(\cdot)$.} According to our distributional assumptions, we have
	\[
	\begin{split}
		m_0 \mid - &\sim \mathrm N(m_1^*, s_1^{2*}),\\
		k_0 \mid - &\sim \textsc{Inv-Gamma}(a_1^*, b_1^*),\\
		b_0 \mid - &\sim \textsc{Gamma}(c_1^*, d_1^*),\\
	\end{split}
	\]
	where we have set
	\begin{align*}
		s_1^{2*} &= \left( \frac{1}{s_1^2} + k_0 \sum_{\ell= 1}^k \frac{1}{\sigma_\ell^{2*}} \right)^{-1},\qquad
		&&m_1^* = s_1^{2*} \left( \frac{m_1}{s_1^2} + k_0 \sum_{\ell = 1}^k \frac{\mu_\ell^*}{\sigma_\ell^{2*}}\right),\\
		a_1^* &= a_1 + \frac{k}{2},\qquad
		&&b_1^* \;\,= b_1 + \frac{1}{2}\sum_{\ell = 1}^k \frac{(\mu_\ell^* - m_0)^2}{\sigma_\ell^{2*}},\\
		c_1^* &= c_1 + k a_0,\qquad
		&&d_1^* \,\,= d_1 + \sum_{\ell=1}^k \sigma_\ell^{2*}.
	\end{align*}
\end{enumerate}

\section{Simulation study with CoRMs}\label{sec:sim_comp}


We have conducted different simulation studies to investigate the impact of different sample sizes and data-generating processes on  posterior inference obtained through Algorithm~\ref{algo:compound}. 
To this end, we sampled sets of synthetic data form the following marginal distributions to obtain six groups of observations: 
\[
\begin{split}
	Y_{1} &\simiid \frac{1}{2} \textsc N (\mu_1, 0.6) + \frac{1}{4} \textsc N (\mu_2, 0.6) + \frac{1}{4} \textsc N (\mu_3, 0.6),\\
	Y_{2} &\simiid \frac{1}{4} \textsc N (\mu_2, 0.6) + \frac{1}{2} \textsc N (\mu_3, 0.6) + \frac{1}{4} \textsc N (\mu_4, 0.6),\\
	Y_{3} &\simiid \frac{1}{2} \textsc N (\mu_1, 0.6) + \frac{1}{2} \textsc N (\mu_2, 0.6),\\
	Y_{4} &\simiid \frac{1}{5} \textsc N (\mu_3, 0.6) + \frac{1}{5} \textsc N (\mu_4, 0.6) + \frac{1}{5} \textsc N (\mu_5, 0.6) + \frac{2}{5} \textsc N (\mu_6, 0.6),\\
	Y_{5} &\simiid \frac{1}{2} \textsc N (\mu_5, 0.6) + \frac{1}{2} \textsc N (\mu_3, 0.6),\\
	Y_{6} &\simiid \frac{1}{4} \textsc N (\mu_1, 0.6) + \frac{1}{4} \textsc N (\mu_2, 0.6) + \frac{1}{4} \textsc N (\mu_3, 0.6) + \frac{1}{4} \textsc N (\mu_5, 0.6),
\end{split}
\]
where some of the components are shared across different groups, meaning they have the same mean value, while others are group-specific. We consider a total of four distinct scenarios for the simulation study, combining two different group-specific sample sizes $n_j \in \{50, 200\}, \; j = 1, \dots, 6$, with two possible parameter choices for $\mu_1, \dots, \mu_6$, i.e.
\[
\begin{matrix}
	\text{(1) } \mu_1 = 6, &\mu_2 = 10, &\mu_3 = 15, &\mu_4 = 20, &\mu_5 = 0, &\mu_6 = 3,\\
	\multicolumn{6}{c}{\text{or}}\\
	\text{(2) } \mu_1 = 4, &\mu_2 = 6.66, &\mu_3 = 10, &\mu_4 = 13.33, &\mu_5 = 0, &\mu_6 = 2.
\end{matrix}
\]
It is worth noticing that specification (1) represents the case where the components are well separated, making it easier to estimate the latent clusters in the data. On the other hand, the second parameter choice (2) corresponds to lower separation of the components. Each scenario has been replicated $100$ times.

We use a dependent mixture model with Gaussian kernels, namely 
\begin{equation*}
	\begin{split}
		Y_{i,j}\mid \mu_{i,j},\sigma_{i,j}^2 &\simiid \textsc N( \mu_{i,j}, \sigma_{i,j}^2),\qquad i=1,\dots,n_j, \;j=1,\dots,d\\
		(\mu_{i,j}, \sigma_{i,j}^2) \mid \tilde p_j &\simind \tilde p_j, \qquad i=1,\dots,n_j, \;j=1,\dots,d\\
		(\tilde p_1, \dots, \tilde p_d) &\sim \Qcr, 
	\end{split}  
\end{equation*} 
where $\Qcr$ is the distribution of a  normalized vector of CoRMs with gamma scores and stable directing L\'evy intensity as in \eqref{eq:choice_f_nu}. The model specification is completed by setting a normal-inverse-gamma distribution $\textsc{NIG}(m_0, k_0, a_0, b_0)$ as the base measure $\alpha(\cdot)$
of $\Qcr$. We further relax the model specification by considering hyperprior distributions on the main parameters, namely $m_0 \sim \textsc N(0, 10)$, $k_0 \sim \textsc{Gamma}(2,2)$, $a_0 = 2$ and $b_0 \sim \textsc{Gamma}(2,2)$, $\sigma \sim \textsc{Beta}(2,2)$, and $\phi = 1$. The specific choice $\phi = 1$ implies that the relative variability (coefficient of variation) of the scores equals $1$, which may be considered an intermediate value.
We ran Algorithm \ref{algo:compound} for $15\,000$ iterations, discarding the first $10\,000$ iterations. 

We assessed the accuracy of the clustering using the normalized variation of information distance \citep{Mei07} between the posterior point estimate of the latent partition in the data and the true one. The point estimates were obtained using a decisional approach based on the variation of information loss function \citep{Wad18, Ras18}. Further, we used the J-divergence \citep{Jef48}, averaged with respect to the different groups of data, to quantify how closely the estimated density functions match the true ones.

\begin{figure}[!ht]
	\centering
	\includegraphics[width = 0.69\textwidth]{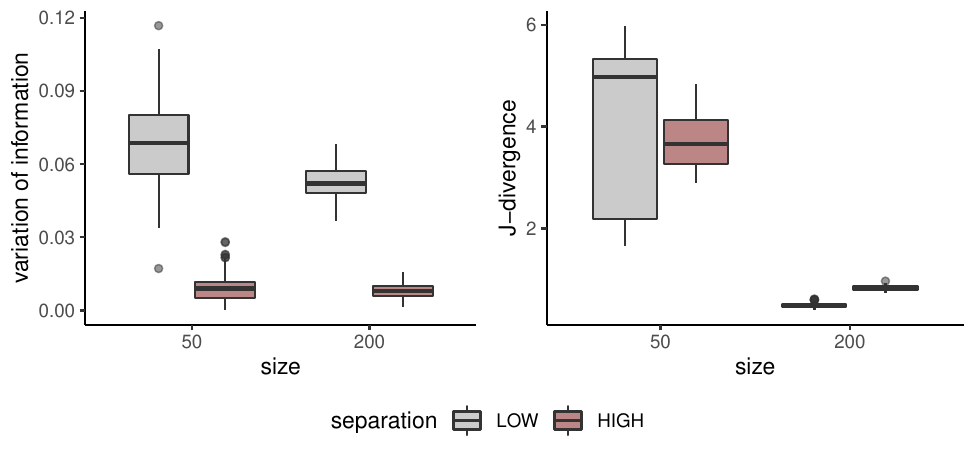}
	\caption{
		Left panel: variation of information distance between true and point estimate of the latent partition of the data. Right panel: J-divergence between data generating and posterior
		point estimate density functions. 
		Different scenarios
		are a combination of low and high number of groups (LG and HG respectively) with low and high separation of the components (LS and HS respectively).
		}\label{fig:comp_simu}
	
\end{figure}
Figure \ref{fig:comp_simu} shows the box plots for the normalized variation of information and the J-divergence for all the scenarios considered in the study. We can observe that the posterior estimates of the latent partitions are more accurate when the separation of data-generating processes increases, as the variation of information distance is dropping toward zero. This means that, as the location parameters $\mu_1, \dots, \mu_6$ become further apart from each other, the posterior estimates of the latent partitions are getting closer to the true values. Additionally, the posterior estimates of the random densities become closer to the true values as the group-specific sample sizes increase. 

%
\section{Analysis of the Sloan Digital Sky Survey data}\label{app:sky}

We consider a dataset of $n = 24\,312$ galaxies, drawn from the Sloan Digital Sky
Survey first data release \citep{Aba03}. The galaxies are naturally partitioned into $d = 25$ different groups, by considering all the cross-classifications of  $5$ kinds of  luminosity and $5$ groups of different environment. Our inferential interest is to properly reconstruct the density function of the difference between ultraviolet and red filter (U-R color), specifically for each group in the study. As highlighted by \citet{Bal04}, the bimodal structure of such density functions has been established. Since some groups contain few observations to properly recognize this behavior, the dependence structure induced by CoRMs allow a  borrowing of information across the different groups, especially useful for groups having few data.  We refer to \citet{Can19, Ste21} for allied approaches. 

We consider the hierarchical mixture model \eqref{eq:mod_mixture} with Gaussian kernels,  where the mixing measures are normalized CoRMs with gamma distributed scores, stable directing L\'evy measure, and centering measure equal to a normal-inverse-gamma distribution, as done in Section~\ref{sec:sim_comp}. The model specification is completed by assuming the following hyperpriors on the main parameters of the model: $m_0 \sim \textsc N(0, 10)$, $k_0 \sim \textsc{Gamma}(2,2)$, $a_0 = 2$, $b_0 \sim \textsc{Gamma}(2,2)$, $\sigma \sim \textsc{Beta}(2,2)$ and $\phi \sim \textsc{Gamma}(2,2)$. We perform posterior inference by using Algorithm \ref{algo:compound}. See also Section~\ref{app:algo} for further details. We ran the algorithm for $15\,000$ iterations, including $10\,000$  burn-in iterations.  

\begin{figure}[!ht]
	\centering
	\includegraphics[width = 0.89\textwidth]{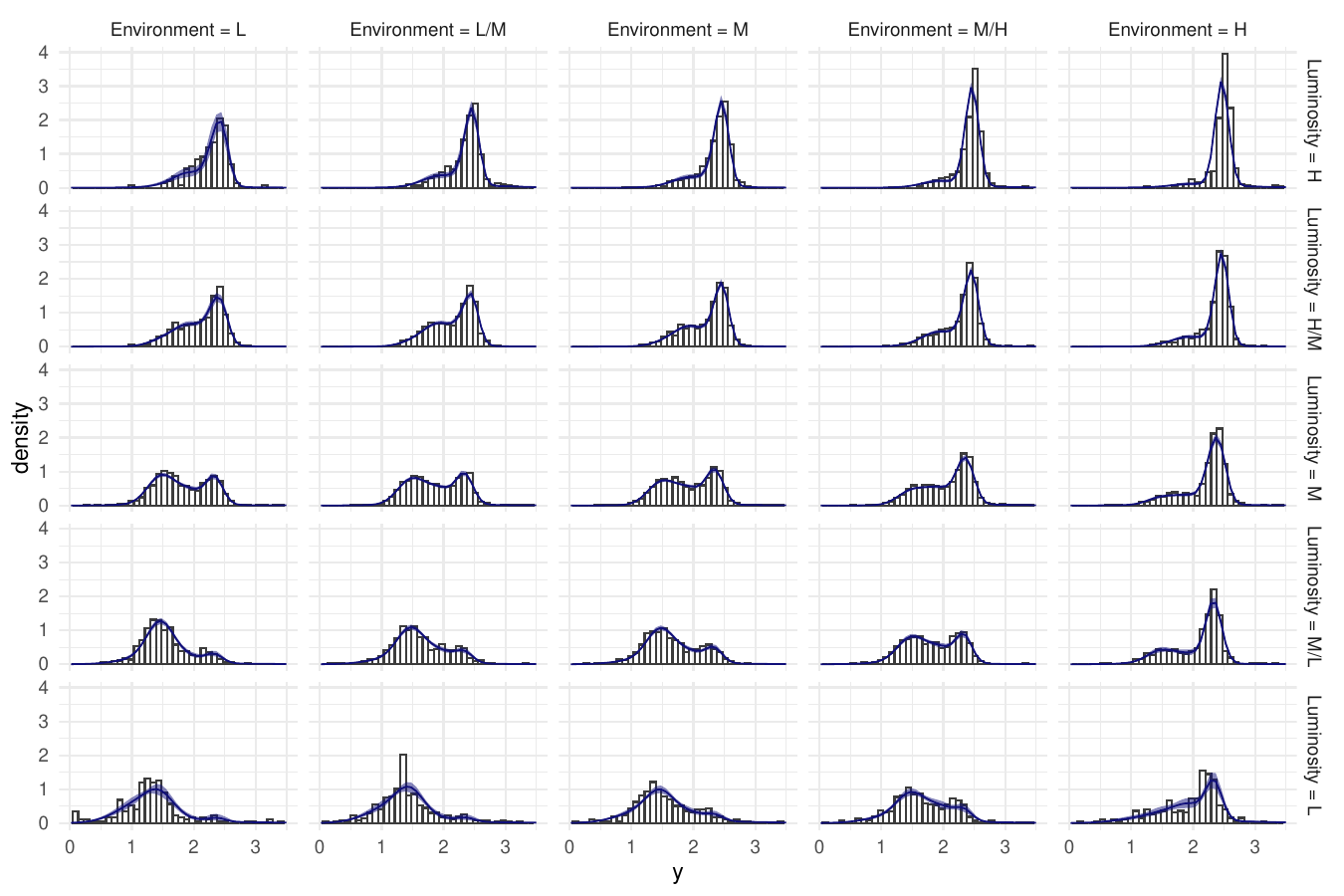}
	\caption{Posterior density estimates of the Sloan Digital Sky Survey data. Different columns correspond to different groups of environment. Different rows correspond to different groups of luminosity. Each panel show the histogram of the observed data, the estimated posterior mean of density (blue line) with a $0.95$-probability credible band (blue filled area). }\label{fig:sloan}
\end{figure}

The group-specific density estimates are reported in Figure \ref{fig:sloan}.
We can appreciate how the density estimates for the groups corresponding to low or low/medium level of both environment and luminosity feature  bimodal distributions with a positive skewness, i.e., the mass is concentrating on the left side of the support. On the counterpart, groups with high or medium/high level of both the grouping variables tend to have negative skewness in their distribution, with the mass concentrating on the right side of the support. 

One advantage of using CoRMs relies on the possibility of deriving the baseline measure, that is the latent shared measure $\tilde p=\tilde \eta/ \tilde \eta (\X)$, which is transformed in group-specific mixing measures $\tilde p_1, \dots, \tilde p_d$. We can then define a latent baseline mixture model as
\begin{equation} \label{eq:directing_mixture}
	\tilde f_0(y) = \int \kernel(y; (\mu, \sigma^2)) \tilde p(\D\mu, \D\sigma^2).
\end{equation}
Note that $\tilde f_0$ represents the information shared across different groups. 
The posterior mean of the baseline mixture model, shown in Figure~\ref{fig:trace_sloan}, clearly shows the aforementioned  bimodal behaviour, which is then specialized in the group-specific mixture models of Figure~\ref{fig:sloan}, where the weights of the common mixture $\tilde f_0$ are adjusted by the group-specific~scores.


\begin{figure}
	\centering
	\includegraphics[width=0.52\textwidth]{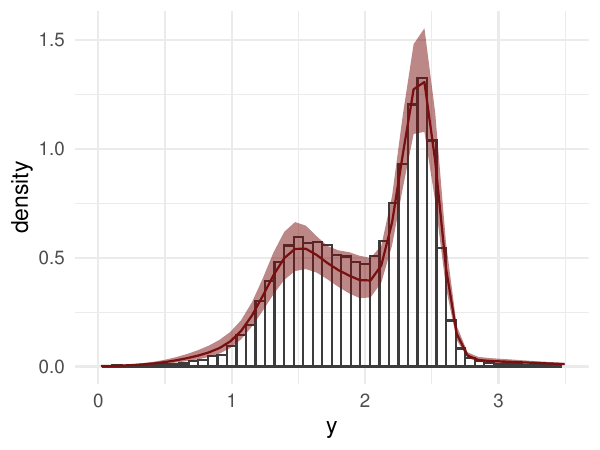}
	\caption{
		Posterior mean of the directing mixture model \eqref{eq:directing_mixture}. The histogram refers to all the data reported in the Sloan Digital Sky Survey. The red line is the estimated posterior mean of the density, while the red filled area is the $0.95$-probability credible band.}\label{fig:trace_sloan}
\end{figure}


\end{document}